\documentclass[letterpaper,11pt]{article}
\usepackage[margin=1in]{geometry}
\usepackage{graphicx}
\usepackage{lscape}
\usepackage{subfigure}
\usepackage{amssymb}
\usepackage{epstopdf}
\usepackage{graphicx}
\usepackage[tbtags]{amsmath}
\usepackage{amssymb}
\usepackage{amsthm}
\usepackage{amsmath}
\usepackage{multicol}
\usepackage{array,multirow}
\usepackage{pgfplots}
\usepackage[titletoc]{appendix}
\usepackage{natbib}
\usepackage{authblk}

\usetikzlibrary{patterns}
\newtheorem{theorem}{Theorem}[section]
\newtheorem{proposition}{Proposition}[section]
\newtheorem{lemma}{Lemma}[section]

\newcommand{\ab}{\left[ \alpha/\beta \right]}

\newcommand{\abt}{\left[ \alpha/\beta \right]_{\mathrm{T}}}

\newcommand{\abn}{\left[ \alpha/\beta \right]_{\mathrm{N}}}

\newcommand{\ma}{\xi}

\newcommand{\mat}{\xi_{\mathrm{T}}}

\newcommand{\man}{\xi_{\mathrm{N}}}

\newcommand{\ca}{\theta}

\newcommand{\cat}{\theta_{\mathrm{T}}}

\newcommand{\can}{\theta_{\mathrm{N}}}

\newcommand{\en}{\bar{b}}

\newcommand{\sfe}{\delta_{\mathrm{eff}}}

\DeclareMathOperator*{\argmax}{arg\,max}

\title{A Mathematical Programming Approach to the Fractionation Problem in Chemoradiotherapy}
\date{}
\author[1]{Ehsan Salari}
\author[2]{Jan Unkelbach}
\author[2]{Thomas Bortfeld}
\affil[1]{Department of Industrial and Manufacturing Engineering, Wichita State University, Wichita, KS, 67260}
\affil[2]{Department of Radiation Oncology, Massachusetts General Hospital and Harvard Medical School, Boston, MA, 02114}

\begin{document}
\maketitle
\abstract{In concurrent chemoradiotherapy, chemotherapeutic agents are administered during the course of radiotherapy to enhance the primary tumor control. However, that often comes at the expense of increased risk of normal-tissue complications. The additional biological damage is mainly attributed to two mechanisms of action, which are the independent cytotoxic activity of chemotherapeutic agents and their interactive cooperation with radiation. The goal of this study is to develop a mathematical framework to obtain drug and radiation administration schedules that maximize the therapeutic gain for concurrent chemoradiotherapy. In particular, we analyze the impact of incorporating these two mechanisms into the radiation fractionation problem. Considering each mechanism individually, we first derive closed-form expressions for the optimal radiation fractionation regimen and the corresponding drug administration schedule. We next study the case in which both mechanisms are simultaneously present and develop a dynamic programming framework to determine optimal treatment regimens. Results show that those chemotherapeutic agents that interact with radiation may change optimal radiation fractionation regimens. Moreover, administration of chemotherapeutic agents possessing both mechanisms may give rise to optimal non-stationary fractionation schemes.}\\

\noindent \textbf{Keywords:} Chemoradiotherapy, fractionation, biologically-effective dose, mathematical optimization, dynamic programming

\section{Introduction} \label{sec: introduction}
In radiotherapy, high-energy radiation is used to eradicate cancer cells by damaging their deoxyribonucleic-acid (DNA) molecule. As the radiation beam passes through the patient it damages both cancerous and normal cells along its path. The majority of radiotherapy treatment plans are divided into daily treatment \emph{fractions} and delivered over the course of one to several weeks. This allows \emph{late-responding} normal tissues to partially \emph{repair} the radiation-induced damage between treatment fractions. Moreover, it is suggested that the core of large tumors may be oxygen deprived, so-called \emph{hypoxic}, and thus more radio-resistant \citep{Coleman1988,Kumar2000}. Fractionated radiotherapy allows for \emph{reoxygenation} of hypoxic tumor regions during the course of the treatment, thereby making the tumor cells more susceptible to radiation. Radio-sensitivity of a cell also depends on the phase of the cell cycle. Cells in their S phase and G$_2$-M phases are relatively more radio-resistant and radio-sensitive, respectively \citep{Terasima1961}. Hence, fractionated radiotherapy benefits from \emph{redistribution} of surviving tumor cells in radio-sensitive cell-cycle phases. Finally, tumor and normal-tissue \emph{repopulation} during the course of the treatment is an important factor considered in fractionated radiotherapy (see \cite{withers1975,Pajonk2010} for 4 R's of fractionation, which are repair, reoxygenation, redistribution, and repopulation). Therefore, the choice of fractionation \emph{regimen}, that is, the number of fractions used and the radiation dose delivered per fraction, can play an important role in the treatment outcome. The \emph{biologically-effective dose} (BED) model is a commonly-used concept in clinical practice to measure and compare the biological damage caused by different radiation fractionation regimens (see \cite{Fowler89,Fowler2010,Hall2006}).

For many advanced-stage cancers, radiotherapy alone is insufficient to successfully eradicate all tumor cells without severely damaging the surrounding normal tissues. In particular, for inoperable tumors, \emph{chemoradiotherapy} (CRT) is the standard of care, in which one or several chemotherapeutic agents are administered along with radiation to enhance tumor cell kill. CRT treatment regimens are classified into neoadjuvant, concurrent, or adjuvant depending on whether the chemotherapeutic agents are administered before, during, or after the course of radiotherapy, respectively. The rationale for combining chemotherapy with radiation in concurrent CRT may vary across different chemotherapeutic agents. Nevertheless, \cite{Steel1979} proposed a conceptual framework to summarize possible ways of combining the two modalities. In the following, we discuss an adaptation of this framework suggested by \cite{bernier2003}:
\begin{itemize}
	\item \emph{spatial cooperation} referring to a non-interactive cooperation in which radiation eradicates the primary local tumor while chemotherapy targets the systemic disease and sub-clinical metastases
	\item \emph{additivity} referring to added locoregional effects of the two modalities assuming no interaction between the cytotoxic activity of radiation and chemotherapeutic agents 
	\item \emph{radio-sensitization} referring to enhancement of radiation-induced cell kill due to chemotherapeutic agents
	\item \emph{radio-protection} referring to an antagonistic effect of chemotherapeutic agents on radiation-induced damage.
\end{itemize}
Many clinical trials have demonstrated lower incidence rate of distant metastases in patients treated with concurrent CRT compared to radiotherapy alone, which verifies the spatial cooperation mechanism \citep{Seiwert2007a}. Moreover, it has been clinically shown that concurrent CRT leads to a better local tumor control for several treatment sites compared to radiotherapy alone (see, e.g., \cite{Seiwert2007a}). However, that often comes at the expense of increased risk of normal-tissue complications (see, e.g., \cite{Parashar2011}). The enhancement of the local tumor control is explained by the additivity and radio-sensitization mechanisms. More specifically, the additive cytotoxic activity of chemotherapeutic agents is attributed to the drug-induced DNA damage. For instance, Cisplatin is a platinum-based anti-cancer drug, which is commonly used in concurrent CRT. It is suggested that Cisplatin distorts the DNA molecule structure via introducing inter- and intra-strand cross-links leading to cell death \citep{Huang1995}. Radio-sensitization is the mechanism through which the radiation-induced damage is enhanced, which also has been referred to as \emph{supra-additivity} or \emph{synergism} in the literature. Several mechanisms of action have been suggested for radio-sensitization including (i) direct enhancement of initial radiation damage by incorporating drugs into the DNA, (ii) inhibition of post-irradiation DNA-damage repair by blocking the corresponding pathways, (iii) accumulation of cells in a radio-sensitive phase (G$_2$ and M phases) or elimination of cells that are in a radio-resistant phase (S phase), (iv) elimination of hypoxic cells, and (v) inhibition of tumor repopulation (see \cite{nishimura2004,Seiwert2007a} for a detailed discussion). For instance, radio-sensitization properties of 5-Fluorouracil (5-FU), a chemotherapeutic agent typically administered intravenously during radiotherapy, are primarily ascribed to mechanisms (i)--(iii) \citep{Seiwert2007a}. Finally, the radio-protection mechanism ideally allows larger radiation doses by selectively protecting the normal tissue against radiation. An example of a chemotherapeutic agent with radio-protection properties is Amifostine, which is used to reduce the incidence or severity of normal-tissue toxicity without compromising the local tumor control \citep{Brizel2000,Movsas2005}; however, the radiation protection benefit of Amifostine remains disputable in clinical trials \citep{Buentzel2006}. This paper focuses on concurrent CRT regimens that are founded upon additivity and radio-sensitization mechanisms.

The additivity and radio-sensitization mechanisms are rarely specific to tumor cells and may also increase the biological damage to normal tissues \citep{nishimura2004}. Therefore, a therapeutic benefit is only achieved if enhancement of the primary tumor control outweighs the additional normal-tissue toxicity. Moreover, the interaction between the drug and radiation, particularly the radio-sensitization mechanism, is time and dose dependent. Hence, combining the best available radiotherapy fractionation regimen with the most promising chemotherapy drug-administration schedule may not necessarily optimize the treatment outcome for concurrent CRT \citep{bernier2003}. To maximize the therapeutic gain for concurrent CRT one has to simultaneously consider the radiation and drug administration schedules. This suggests the possible role of mathematical optimization in assisting with the design of promising CRT treatment regimens using dose-response models. However, many of the treatment regimens currently used in concurrent CRT have been solely initiated from phase I and II clinical trials rather than the underlying biological principles \citep{Wilson2006}. Although mathematical optimization has been vastly employed to determine fractionation regimens in radiotherapy (see, e.g., \cite{ Ramaknishnan2013, Wein2000, Yang2005}) and drug administration schedules in chemotherapy (see \cite{Shi2011} for a complete review), the application of mathematical optimization techniques to CRT fractionation decision has not been adequately explored. To the best of our knowledge, \cite{Jones2005} is the only study that promotes the use of mathematical modeling and optimization for CRT fractionation decision. The goal of this paper is to take an initial step toward developing a mathematical framework to study the fractionation problem in concurrent CRT. The aim, in particular, is to identify changes in optimal fractionation regimens that result from adding chemotherapeutic agents to the radiation treatment. To achieve this goal, we extend the radiation BED model to account for the additivity and radio-sensitization mechanisms. Using the extended BED model we then formulate the CRT fractionation problem as a non-linear programming model. We derive closed-form expressions for the optimal fractionation regimens for special cases in which either additive or radio-sensitization effect is relevant. Furthermore, we develop a \emph{dynamic programming} (DP) framework to study optimal fractionation regimens for chemotherapeutic agents that possess both mechanisms (see \cite{Bertsekas2000} for a complete review of DP).

The remainder of this paper is organized as follows: In Section \ref{sec: fractionation}, we introduce the BED model and discuss how it is used to make fractionation decision in radiotherapy. We then consider the extension of the BED model to CRT and formulate the CRT fractionation problem. In Section \ref{sec: solution}, we derive closed-form solutions to the CRT fractionation problem considering the additive and radio-sensitization effects individually. We also develop a DP algorithm to solve the CRT fractionation problem when both effects are simultaneously present. In Section \ref{sec: example}, we apply the results derived in Section \ref{sec: solution} to a photon and proton treatment plan and study the impact of introducing chemotherapeutic agents on the optimal fractionation regimens. In Section \ref{sec: discussion}, we present some insights into the CRT fractionation decision provided by the obtained results and discuss limitations of the study and future research directions. Finally, in Section \ref{sec: conclusion}, we summarize and conclude the paper.

\section{Fractionation decision} \label{sec: fractionation}
 In this section we review the optimal fractionation regimens for radiotherapy, discuss the fractionation decision in the context of chemoradiotherapy, and formulate the corresponding fractionation problem.

\subsection{Radiotherapy}\label{sec: rtfx}
 It is widely accepted that the logarithm of the surviving fraction of irradiated cells follows a linear-quadratic fit \citep{Hall2006}. Originally motivated by this observation, the notion of BED is defined and used in clinical practice to quantify and compare the biological effect of radiation fractionation regimens. The BED associated with a fractionation regimen with $n$ treatment fractions, where a radiation dose of $d_i$ Gray (Gy) is administered at fraction $i=1\ldots,n$, is 
\[ \mathrm{BED}=\sum_{i=1}^n d_i\left(1 + \frac{ d_i}{\ab}\right), \]
where $\ab$ (measured in Gy) is a tissue-specific parameter. Several clinical studies have associated different BED values, independent of the number of fractions $n$ and radiation doses $d_i \; (i=1,\ldots,n)$, with \emph{tumor control probability} (TCP) and \emph{normal-tissue complication probability} (NTCP) and have suggested BED tolerance values for different normal tissues to avoid radiation toxicity \citep{Marks2010}. In order to achieve the maximum TCP in many treatment sites (e.g., lung cancer), it is clinically desirable to determine fractionation regimens that deliver the maximum BED in the target region while ensuring that the BED in the surrounding normal tissue does not exceed the corresponding tolerance value. Several studies have used mathematical optimization to determine optimal fractionation regimens to achieve this. \cite{Mizuta2012} studied the dependence of the optimal fractionation regimen on the BED parameters of the target $\abt$ and normal tissue $\abn$ assuming each receives a uniform dose of $d_i$ and $\delta d_i$ ($\delta \geq 0$ is so-called \emph{sparing factor}) at fraction $i=1,\ldots,n$, respectively. In particular, they showed that if $\abn > \delta \abt$, then a stationary \emph{hypo-fractionation} regimen is optimal, that is, using as few fractions as possible with equal doses per fraction; otherwise, a stationary \emph{standard fractionation} regimen is preferable, that is, using the maximum number of fractions allowed with equal doses per fraction. Lastly, using extended BED models \cite{Ramaknishnan2013} and \cite{Yang2005} incorporated the tumor repopulation effect into the fractionation decision. 

The assumption of the uniformity of the normal-tissue dose distribution is unrealistic in clinical practice. Radiotherapy plans often deliver a uniform dose to the target region while depositing a heterogeneous dose distribution in the normal tissue, which, in turn, leads to a heterogeneous BED distribution in the normal tissue. The normal-tissue tolerance against radiation depends on the arrangement of the \emph{functional subunits} (FSUs) in the organ. FSU is the largest unit of cells capable of being regenerated from a single cell without any functionality loss \citep{Withers1988}. In \emph{parallel} organs, such as lung, FSUs are arranged in parallel and function relatively independently. As a result, parallel structures are able to tolerate a relatively large radiation dose if it is limited to small sub volumes. On the contrary, in \emph{serial} organs, such as spinal cord, FSUs are arranged in series and the integrity of each FSU is critical to the organ function. Hence, large radiation is harmful to serial organs even if it is delivered to a small sub volume. In radiotherapy treatment planning, this effect is typically accounted for by using the generalized mean of the dose distribution \citep{niemierko1999}. In particular, for a ``perfectly'' parallel and serial organ the generalized mean reduces to the average and maximum of the dose distribution, respectively. Using the generalized mean of the BED distribution, \cite{keller2013} and \cite{Unkelbach2013} have independently extended the fractionation decision studied in \cite{Mizuta2012} to allow for a heterogeneous dose distribution in the normal tissue. In particular,  they have identified ranges of the ratio $\abn/\abt$ for which a hypo- or standard fractionation regimen is optimal. In this paper, we extend the framework developed in \cite{Unkelbach2013} to study the CRT fractionation decision.


\subsection{Chemoradiotherapy}\label{sec: crtfx}
Chemotherapeutic agents administered in concurrent CRT cause additional biological damage in the target and normal tissue owing to the additivity and radio-sensitization mechanisms discussed in Section \ref{sec: introduction}. Several dose-response models have been proposed in the literature to study this additional biological damage. We first review those models and then formulate the CRT fractionation problem assuming a dose-response model motivated by those studies. 

\subsubsection{CRT dose-response model}\label{sec: crtfx-dose-response}
Steel and Peckham introduced the \emph{isobologram} analysis to study the interaction between chemotherapeutic agents and radiation \citep{steel1979b,Steel1979}. More specifically, an isobologram is generated by plotting different drug and radiation dose levels that lead to an identical cytotoxic effect, which are called \emph{isoeffect} curves. Two isoeffect curves corresponding to two different modes of interaction between the drug and radiation are generated. These curves form the-so-called \emph{envelope of additivity}. Next, the true isoeffect curve is empirically obtained and contrasted against the envelope. Depending on whether this curve lies above, within, or below the envelope, the drug and radiation interaction is classified as sub-additive (antagonistic), additive, or supra-additive (synergistic), respectively. Note that the isobologram analysis may only be used to classify the type of interaction between a given chemotherapeutic agent and radiation, and it does not provide an explicit dose-response relationship for concurrent CRT. However, to incorporate the additivity and radio-sensitization mechanisms in the CRT fractionation decision, one needs to employ a quantitative model that associates different drug and radiation dose levels with biological damage. One possible approach is to extend the radiation BED model to account for the two mechanisms. The earliest extension in the literature is the introduction of the \emph{dose modifying factor} (DMF), also known as the \emph{sensitizer enhancement ratio} (SER), to account for the additional biological damage caused by the chemotherapeutic agent. It is defined as the ratio of radiation doses with and without the chemotherapeutic agent that lead to the same biological effect \citep{wigg2001}. However, the DMF does not distinguish between additivity and radio-sensitization mechanisms and does not consider the drug's concentration. Other studies have used an extended BED model that only accommodates the additivity mechanism \citep{Jones2005,Vogelius2011a,Vogelius2011b}. In particular, they introduce the notion of \emph{chemotherapy-equivalent radiation dose} (CERD), which is an additional radiation BED measured in fractional doses of 2 Gy, to account for the cytotoxic activity of the drug. Although CERD may account for the additivity mechanism, it is intrinsically incapable of modeling the radio-sensitization effect. 

\emph{Clonogenic assay} is an in-vitro technique to determine the impact of a single or a combination of agents on the survival and proliferation of cells \citep{Hall2006}. The graphical representation of this analysis, which illustrates the relationship between the fraction of surviving cells and the drug's concentration or radiation dose, is called a \emph{survival curve}. Survival curves for many malignant and normal cell lines grown in culture and exposed to radiation or chemotherapeutic agents have been determined. In particular, survival curves of irradiated cell lines, regardless of their species of origin, follow a \emph{linear-quadratic} (LQ) fit on a logarithmic scale \citep{Hall2006}.
Several clonogenic-assay studies address the additivity and radio-sensitization mechanisms through studying the changes observed in the LQ survival curve if irradiated cell lines are also exposed to chemotherapeutic agents. In the following, we highlight some of those studies to motivate an extension to the BED model for concurrent CRT.

\paragraph{Additivity mechanism} The additivity mechanism leads to an increase in cell kill due to the cytotoxic activity of the chemotherapeutic agent. The underlying assumption for the additivity mechanism is the lack of any interactive cooperation between the drug and radiation. Thus, individual survival curves of radiation and chemotherapeutic agents can be employed to characterize the survival curve for the combined case. Several clonogenic-assay studies suggest that survival curves of cell lines that are only exposed to chemotherapeutic agents have an exponential form in terms of the drug's concentration (see, e.g., \cite{Berenbaum1969,Eichholtz1980,giocanti1993,Skipper1964}). Hence, to account for the additivity mechanism, one may add a linear term with respect to the drug's concentration to the exponent of the LQ model (see Appendix \ref{sec: appendix_BEDExt}). 

\paragraph{Radio-sensitization mechanism} 
This mechanism enhances the radiation-induced cell kill. Therefore, radio-sensitization may change the shape of the LQ survival curve. In particular, several clonogenic-assay studies suggest that this mechanism increases the initial slope of the LQ model. More specifically, they report an increase in the linear term of the LQ model without an  appreciable change in the quadratic term (see, e.g., \cite{Franken2001,Franken2012,dai2013,joiner2009,Chavaudra1989,Miller1992b}). Moreover, the extent of this increase depends on the drug's concentration. Therefore, to accommodate the radio-sensitization mechanism, one may add a multiplicative term with respect to the radiation dose and the drug's concentration to the exponent of the LQ model (see Appendix \ref{sec: appendix_BEDExt}).

\paragraph{Extending BED to concurrent CRT}
 The radiation BED model is originally motivated by the LQ survival curve. Chemotherapeutic agents change the LQ survival curve due to their additivity and radio-sensitization mechanisms. Hence, to extend the BED model to concurrent CRT we assume that the additional biological damage due to additivity and radio-sensitization is accounted for via adding a linear and multiplicative term to the BED model, respectively. More specifically, consider a treatment regimen in which a drug concentration of $c_i$ and a radiation dose of $d_i$ are administered at treatment fraction $i=1,\ldots,n$. We define a function $B: \mathbb{R}^2 \rightarrow \mathbb{R}$ as
\begin{align} \label{def: BED}
   B \left(d_i,c_i \right) = d_i\left(1 + \frac{d_i}{\ab}\right) +  \ca c_i + \ma c_i d_i,
\end{align}
where $\ca$ and $\ma$ are tissue-specific parameters associated with additive and radio-sensitization effects, respectively. We then use $B$ to measure the total biological damage caused in the tissue due to radiation and drug administration at fraction $i=1,\ldots,n$. In order to calibrate the effect of the additivity and radio-sensitization mechanisms against radiation, the parameters $\ca$ and $\ma$ are given in units of Gy per drug unit and the inverse of drug unit, respectively. Analogous to the case of BED for radiation alone, it is then assumed that a given value of extended BED, independent of the number of fractions $n$ and fractional radiation and drug doses $\left(c_i,d_i\right)\,(i=1,\ldots,n)$, is associated with a certain biological effect. 

\subsubsection{Formulation of the CRT fractionation problem}
We next consider the fractionation decision for concurrent CRT to determine a radiation fractionation regimen along with the drug administration schedule to maximize the BED received by the target while ensuring the BED delivered to the normal tissue does not exceed the corresponding tolerance. At each treatment fraction, the BED in the target and normal tissue is measured using, respectively, functions $B_{\mathrm{T}}$ and $B_{\mathrm{N}}$ with tissue-specific parameters. To formulate the CRT fractionation problem, we adopt the common approach, used in all studies discussed in Section \ref{sec: rtfx} and in particular \cite{Unkelbach2013}, that prior to the fractionation decision, the radiation treatment plan and the associated dose distribution are determined and given. Moreover, we assume that this given treatment plan delivers a uniform dose to the target and a heterogeneous dose distribution to the normal tissue. Therefore, the \emph{spatial} dose distribution in the target and normal tissue is fixed; however, the dose per fraction may be scaled up or down as desired. Let $V$ represent the set of all voxels in the normal tissue. The spatial dose distribution delivered by the given radiation treatment plan can then be characterized as follows: the target receives a uniform radiation dose of $d_i$ at fraction $i=1,\ldots,n$, and voxel $v\in V$ in the normal tissue receives a radiation dose of $\delta_{v} d_i$, where $\delta_v \geq 0$ is the sparing factor of voxel $v$. It is assumed that $\delta_v\,(v \in V)$ does not change during the course of the treatment. Let $n$ represent an upper bound on the total number of treatment fractions allowed, which is enforced due to logistical reasons or to avoid undesirable consequences of prolonging the treatment such as tumor repopulation. The fractionation problem can then be formulated in terms of decision variables of radiation dose $d_i$ and drug concentration $c_i$ at fraction $i=1,\ldots,n$ as follows:
\[\max \; \sum_{i=1}^n  B_{\mathrm{T}} \left(d_i,c_i \right)\]
subject to \hfill{(M)}
\begin{align}
	\frac{1}{|V|} \sum_{v  \in V}\sum_{i=1}^n   B_{\mathrm{N}} \left( \delta_v d_i,c_i \right) &\leq \en   \label{eqn: BED}\\
	 c_i &\leq \bar{c} & i&=1,\ldots,n  \label{eqn: drug}\\
	d_i,c_i &\geq 0 & i &= 1,\ldots,n. \label{eqn: nonneg}
\end{align}
The objective function evaluates the cumulative BED delivered to the target over all treatment fractions. The constraint in (\ref{eqn: BED}) ensures that the mean of the BED distribution in the normal tissue does not exceed the specified threshold. This constraint is derived assuming that the normal tissue has a parallel organ structure. If the normal tissue has a serial structure, then (\ref{eqn: BED}) must be changed to limit the maximum of the BED distribution delivered to normal-tissue voxels. Note that $B_{\mathrm{N}}$ is an increasing function with respect to $\delta_v\,(v\in V)$. Hence, the maximum of the BED distribution over all voxels $v\in V$ corresponds to the voxel with the largest sparing factor. Therefore, in order to account for a serial normal tissue, the set of voxels $V$ in (\ref{eqn: BED}) is replaced with a single voxel with the largest sparing factor. Furthermore, in Appendix \ref{sec: appendix_intermNT} we discuss an approximation problem that allows for the application of the developed framework to general organ structures. Constraints in (\ref{eqn: drug}) limit the drug's concentration beyond $\bar{c}$ to prohibit concentration levels that lead to intolerable distant toxicities. More specifically, most chemotherapeutic agents, in addition to locoregional toxicity in the dose-limiting normal tissue, may have side effects such as nausea, fatigue, and diarrhea. Moreover, they can cause distant toxicity, such as bone marrow suppression, nephrotoxicity (toxicity in kidneys), and ototoxicity (hearing loss), among others (see, e.g., \cite{Seiwert2007b,Seiwert2007a}). Limiting the drug's concentration reduces the risk of drug-induced distant toxicity. Lastly, the set of constraints in (\ref{eqn: nonneg}) ensure nonnegativity of the radiation dose and the drug's concentration at each fraction.

In the next section, we develop solution methods to solve (M) and discuss the structure of optimal fractionation regimens. To develop the solution method, we first consider special cases in which additivity and radio-sensitization mechanisms are considered individually, and derive closed-form expressions for the optimal treatment regimens. We then develop a DP algorithm to solve the case in which both mechanisms are relevant.

\section{Solution method}\label{sec: solution}
To study the optimal regimens to (M), we discuss the following four cases:
\begin{itemize}
	\item[(i)] radiotherapy alone
	\item[(ii)] CRT with only the additivity mechanism
	\item[(iii)]  CRT with only the radio-sensitization mechanism
	\item[(iv)]  CRT with additivity and radio-sensitization mechanisms.
\end{itemize}
We start with presenting the following lemma on the mean of the optimal BED distribution delivered to the normal tissue:
\begin{lemma}\label{lemma: binding}
For any optimal regimen $(d_i^*,c_i^*)\,\left(i=1,\ldots,n\right)$, the mean of the normal-tissue BED distribution assumes its upper-bound value, that is,
 \[\frac{1}{|V|}\sum_{v \in V}\sum_{i=1}^n B_{\mathrm{N}}\left(\delta_vd^*_i,c^*_i\right) = \en \]	
\end{lemma}
\begin{proof}
    $B_{\mathrm{T}}$ and $B_{\mathrm{N}}$ are increasing functions of $d_i$ and $c_i$ for $i=1,\ldots,n$. Therefore, if for a given solution to (M), the constraint on the mean of the BED distribution is not binding, one may increase $d_i$ for some $i$ to obtain a larger objective value until this constraint becomes binding.
\end{proof}
Lemma \ref{lemma: binding} states that any optimal regimen to (M) yields a normal-tissue BED distribution with the maximum mean allowed. Using Lemma \ref{lemma: binding} we can substitute the quadratic term $\sum_{i=1}^n d_i^2$ in $B_{\mathrm{T}}$ and rewrite the objective function of (M) as follows:
\begin{align} \label{eqn: obj}
	 \sum_{i=1}^n B_{\mathrm{T}}\left(d_i,c_i\right) &= \frac{1}{\overline{\delta^2}}\frac{\abn}{\abt}\en + \left(1-\frac{1}{\sfe}\frac{\abn}{\abt} \right)\sum_{i=1}^n d_i \notag\\
&\qquad \qquad  +\left(\frac{\cat}{\can}-\frac{1}{\overline{\delta^2}}\frac{\abn}{\abt} \right) \can \sum_{i=1}^n c_i + \left(\frac{\mat}{\man}-\frac{1}{\sfe}\frac{\abn}{\abt} \right) \man \sum_{i=1}^n c_id_i,
\end{align}
in which $\bar{\delta}$ and $\overline{\delta^2}$ are, respectively, the first and second moments of the sparing factors associated with normal-tissue voxels and $\sfe$ is 
\begin{align*}
	\sfe &\equiv \overline{\delta^2} / \bar{\delta}.
\end{align*}
Using the new expression for the objective function of (M) in (\ref{eqn: obj}) as well as Lemma \ref{lemma: binding}, we derive closed-form solutions to (M) for cases (i)--(iii). The solution methods developed in this section are also applicable to the approximation problem discussed in Appendix \ref{sec: appendix_intermNT}, which accommodates treatment sites in which the normal tissue does not necessarily have a parallel or serial organ structure.

\subsection{Radiotherapy alone}  \label{sec: radiation}
In the absence of chemotherapeutic agents, we have $c_i = 0\;\left(i=1,\ldots,n\right)$. Hence, (\ref{eqn: obj}) reduces to
\begin{align} \label{eqn: obj_rad}
	 \sum_{i=1}^n B_{\mathrm{T}}\left(d_i,0\right) = \frac{1}{\overline{\delta^2}}\frac{\abn}{\abt}\en + \left(1-\frac{1}{\sfe}\frac{\abn}{\abt} \right) \sum_{i=1}^n d_i.
\end{align}
\cite{Unkelbach2013} used (\ref{eqn: obj_rad}) to show that the optimal solution is a stationary fractionation regimen that uses either the minimum or maximum number of fractions allowed, that is, hypo-fractionation or standard fractionation, respectively (see also \cite{Mizuta2012} for the mathematical result on the existence of optimal stationary regimens). We briefly derive this result for later use. Using Lemma \ref{lemma: binding} and considering only stationary solutions, we determine the optimal total radiation dose to be
\begin{align} \label{eqn: D_rad}
	\sum_{i=1}^n d_i = \frac{-1 + \sqrt{1+4\sfe \en / n\bar{\delta}\abn}}{2\sfe/n\abn}.
\end{align}
The expression in (\ref{eqn: D_rad}) is increasing in the number of fractions $n$. Therefore, depending on the coefficient sign for the sum of radiation dose in (\ref{eqn: obj_rad}), which we denote by
\begin{align*}
	\Delta_{\mathrm{r}} = 1-\frac{1}{\sfe}\frac{\abn}{\abt},
\end{align*}
we determine the optimal number of fractions. More specifically, if  $\Delta_{\mathrm{r}} \geq 0$, then the optimal regimen uses the maximum number of fractions allowed (standard fractionation regimen); otherwise, the use of the minimum number of fractions is optimal (hypo-fractionation regimen). Therefore, \cite{Unkelbach2013} introduced the following condition to determine whether a standard fractionation regimen should be used:
\begin{align}
	\abn \leq \sfe \abt. \label{condition: RT}
\end{align}
If (\ref{condition: RT}) is false, then a hypo-fractionation regimen is optimal. 

\subsection{CRT with additivity mechanism} \label{sec: additive}
Several cytotoxic agents, such as Temozolomide, are suggested to only have an additive effect \citep{Chalmers2009}. This corresponds to the chemotherapeutic agent exhibiting insignificant or no radio-sensitization. In other words, $\ma= 0$ in (\ref{def: BED}) for both the target and normal tissue. To derive the optimal solution to (M), we consider two possible scenarios concerning whether the chemotherapeutic agent can solely deposit a BED of $\en$ or larger in the normal tissue, which is summarized by
\begin{align}
B_{\mathrm{N}}\big(0,\bar{c}\big) \geq \frac{\en}{n}. \label{condition: add}
\end{align}
We determine the optimal solution to (M) under the assumption that the condition in (\ref{condition: add}) is satisfied and then generalize the solution to the second scenario in which this condition is not satisfied. Substituting $\mat,\man= 0$ in (\ref{eqn: obj}), we obtain the following expression for the objective function of (M):
\begin{align}\label{eqn: obj_additive}
	 \sum_{i=1}^n B_{\mathrm{T}}\left(d_i,c_i\right) = \frac{1}{\overline{\delta^2}}\frac{\abn}{\abt}\en + \left(1-\frac{1}{\sfe}\frac{\abn}{\abt} \right)\sum_{i=1}^n d_i +\left(\frac{\cat}{\can}-\frac{1}{\overline{\delta^2}}\frac{\abn}{\abt} \right) \can \sum_{i=1}^n c_i.
\end{align}
For any drug concentrations $c_i \;(i=1,\ldots,n)$, it is easy to see that the fractionation problem reduces to the radiation-only case of Section \ref{sec: radiation}, in which the upper bound on the average BED in the normal tissue reduces to $\en - \can \sum_{i=1}^n c_i $. It immediately follows that for CRT with only the additivity mechanism, there is an optimal stationary radiation fractionation regimen. To characterize this regimen, let $b = \can \sum_{i=1}^n c_i$ represent the total BED delivered to the normal tissue by the chemotherapeutic agent. Considering only stationary radiation fractionation regimens with a fixed number of fractions $n$ and using Lemma \ref{lemma: binding}, we determine the optimal total radiation dose, denoted by $D_{\mathrm{a}}=\sum_{i=1}^n d_i$, as a function of $b$ as follows:
\begin{align}\label{eqn: D_additive}
	D_{\mathrm{a}}\left(b;n \right) = \frac{-1 + \sqrt{1+4\sfe\left(\en - b\right) / n\bar{\delta}\abn}}{2\sfe/n\abn}.
\end{align}
The expressions for the objective function of (M) in (\ref{eqn: obj_additive}) and the optimal total radiation dose in (\ref{eqn: D_additive}), lead to the following theorem that describes the optimal solution to (M) for CRT with only the additivity mechanism:
\begin{theorem}\label{theorem: additive}
	Let
\begin{align*}
	\Delta_{\mathrm{r}} &= 1-\frac{1}{\sfe}\frac{\abn}{\abt}, \\
	 \Delta_{\mathrm{a}} &= \frac{\cat}{\can}-\frac{1}{\overline{\delta^2}}\frac{\abn}{\abt}, \\
	 \varrho_{\mathrm{a}} &= \frac{\en}{ D_{\mathrm{a}}\left(0;1\right)}, \\
	 \underline{\rho}_{\mathrm{a}} &=  \bar{\delta}\left[\frac{2\sfe}{n \abn} D_{\mathrm{a}}\left(\en;n\right)+1\right],\\
	 \overline{\rho}_{\mathrm{a}} &= \bar{\delta}\left[\frac{2\sfe}{n \abn} D_{\mathrm{a}}\left(0;n\right)+1\right], \mbox{ and }\\
	 b_{\mathrm{a}}&= \en + \frac{n\abn}{4 \overline{\delta^2}} \left(\bar{\delta}^2 - \left(\Delta_{\mathrm{r}}/\Delta_{\mathrm{a}}\right)^2 \right).
\end{align*}
In the presence of only the additivity mechanism, that is, $\mat = \man = 0$, the optimal treatment regimen to (M) is:
\begin{align*}
\mbox{optimal regimen}:
\begin{cases}	
	\mbox{RT-hypo} & \Delta_{\mathrm{r}} \leq 0,\Delta_{\mathrm{a}} \leq 0, \Delta_{\mathrm{r}} / \Delta_{\mathrm{a}} \leq \varrho_{\mathrm{a}}\\
	\mbox{RT-std} & \Delta_{\mathrm{r}} \geq 0,\Delta_{\mathrm{a}} \leq 0 \\
	\mbox{RT-std} &  \Delta_{\mathrm{r}} \geq 0,\Delta_{\mathrm{a}} \geq 0, \Delta_{\mathrm{r}} / \Delta_{\mathrm{a}} \geq \overline{\rho}_{\mathrm{a}}\\
	\mbox{CRT-std} &  \Delta_{\mathrm{r}} \geq 0,\Delta_{\mathrm{a}} \geq 0, \underline{\rho}_{\mathrm{a}} < \Delta_{\mathrm{r}} / \Delta_{\mathrm{a}} < \overline{\rho}_{\mathrm{a}}\\
	\mbox{CT} &  \Delta_{\mathrm{r}} \geq 0,\Delta_{\mathrm{a}} \geq 0, \Delta_{\mathrm{r}} / \Delta_{\mathrm{a}} \leq \underline{\rho}_{\mathrm{a}}\\
	\mbox{CT} & \Delta_{\mathrm{r}} \leq 0,\Delta_{\mathrm{a}} \geq 0 \\
	\mbox{CT} &  \Delta_{\mathrm{r}} \leq 0,\Delta_{\mathrm{a}} \leq 0, \Delta_{\mathrm{r}} / \Delta_{\mathrm{a}} > \varrho_{\mathrm{a}},
\end{cases}
\end{align*}
where CT, RT-hypo, RT-std, and CRT-std are abbreviated forms of chemotherapy, radiotherapy using hypo-fractionation, radiotherapy using standard fractionation, and chemoradiotherapy using standard fractionation, respectively. Furthermore, CRT-std delivers an average BED of $b_{\mathrm{a}}$ and $\en - b_{\mathrm{a}}$ to the normal tissue via drug and radiation administration, respectively.
\end{theorem}
The proof is in Appendix \ref{sec: appendix_additive}. Figure \ref{fig: additive} illustrates optimal regimens for different values of $ \Delta_{\mathrm{r}}$ and $\Delta_{\mathrm{a}}$ as specified in Theorem \ref{theorem: additive}.
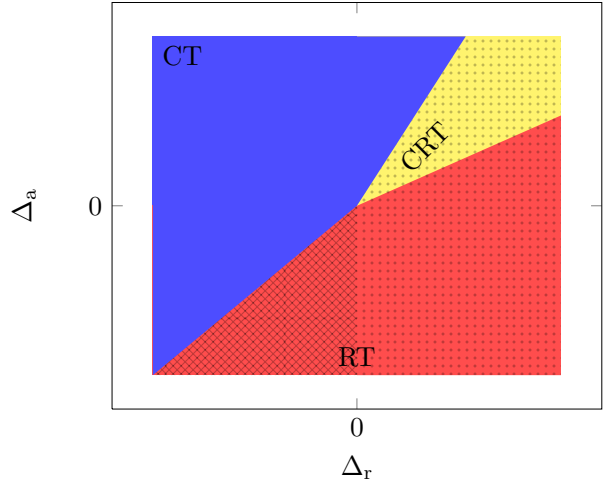
\begin{figure}[h]
\center
\begin{tikzpicture}[scale=0.95]
	\begin{axis}[xtickmin=0,xtickmax=0,ytickmin=0,ytickmax=0, xlabel=$\Delta_{\mathrm{r}}$,ylabel=$\Delta_{\mathrm{a}}$,domain=-15:15,area legend, legend style={anchor=north west}]
		\addplot [fill=yellow!70,draw=none] coordinates {(0,0) (0,15) (15,15) (15,0)} \closedcycle;
		\addplot [fill=blue!70,draw=none] coordinates {(0,0) (0,15) (-15,15) (-15,0)} \closedcycle;
		\addplot [fill=red!70,draw=none] coordinates {(0,0) (0,-15) (-15,-15) (-15,0)} \closedcycle;
		\addplot [fill=red!70,draw=none] coordinates {(0,0) (0,-15) (15,-15) (15,8) (0,0)} \closedcycle;
		\addplot [fill=blue!70,draw=none] coordinates {(0,0) (8,15) (0,15) (-15,15) (-15,-15) (0,0)} \closedcycle;
		\addplot [fill=gray!80,opacity = 0.5,pattern=dots,draw=none] coordinates {(0,-15) (0,0) (8,15) (15,15) (15,-15)  (0,-15)} \closedcycle;
		\addplot [fill=gray!80,opacity = 0.5,pattern=crosshatch,draw=none] coordinates {(0,-15) (0,0) (-15,-15) (0,-15)} \closedcycle;
		\node at (axis cs: -15,15) [anchor=north west] {CT};
		\node at (axis cs:  0,-15) [anchor=south]{RT};
		\node at (axis cs:   2,4) [anchor=north west,rotate=45]{CRT};
    \end{axis}
\end{tikzpicture}
\caption{\label{fig: additive} \small{Schematic of optimal regimens for CRT with only the additivity mechanism. Dotted and hatched regions illustrate whether standard or hypo-fractionated radiotherapy is optimal, respectively.}}
\end{figure}
The following observations are made for optimal treatment regimens associated with CRT with only the additivity mechanism:
\begin{itemize}
	\item[1.] If $\Delta_{\mathrm{r}} \geq 0$, we have $\abn \leq \sfe \abt$, and thus it is optimal to administer radiation using standard fractionation. This fractionation decision is independent of the use of the chemotherapeutic agent.
	\item[2.] If $\Delta_{\mathrm{r}} \geq 0$ and $\Delta_{\mathrm{a}} < 0$, it is optimal to use radiation alone without administering the drug. In contrast, if $\Delta_{\mathrm{r}} \leq 0$ and $\Delta_{\mathrm{a}} > 0$ it is optimal to use chemotherapy alone with no radiotherapy.
	\item[3.] For the case where $\Delta_{\mathrm{r}}$ and $\Delta_{\mathrm{a}}$ have the same sign, the ratio $\Delta_{\mathrm{r}} / \Delta_{\mathrm{a}}$ determines whether the chemotherapeutic agent should be used. Transitioning thresholds are specified by $\underline{\rho}_{\mathrm{a}}$, $\overline{\rho}_{\mathrm{a}}$, and  $\varrho_{\mathrm{a}}$.
	\item[4.] If a standard radiation fractionation is used, then there is a range of $\Delta_{\mathrm{r}} / \Delta_{\mathrm{a}}$, that is, $ \left(\underline{\rho}_{\mathrm{a}},\overline{\rho}_{\mathrm{a}} \right)$, within which radiation and chemotherapy are combined. In contrast, for a hypo-fractionated radiation regimen it is never optimal to administer an additive chemotherapeutic agent concurrently with radiation. This aspect is explained further in Section \ref{sec: discussion}.
\end{itemize}

We next discuss optimal regimens when (\ref{condition: add}) is not satisfied. In other words, the drug alone is not able to deliver a BED of $\en$ or larger to the normal tissue. It is easy to see that those cases in Theorem \ref{theorem: additive} that only use radiation, which are RT-std and RT-hypo, continue to be optimal under the new condition. However, the remaining cases, which are CT and CRT-std, need to be adjusted accordingly. In particular, for the case of drug administration only CT, the chemotherapeutic agent can deliver a maximum BED of $n B_{\mathrm{N}}\big(0,\bar{c}\big)$. Hence, in order to have the normal-tissue BED constraint be binding (see Lemma \ref{lemma: binding}), an additional average BED of $\en- n B_{\mathrm{N}}\big(0,\bar{c}\big)$ is required to be delivered using hypo- or standard fractionation radiation depending on the sign of $\Delta_{\mathrm{r}}$. Finally, for the case of CRT-std, a BED of $\min\left\{n B_{\mathrm{N}}\big(0,\bar{c}\big), b_{\mathrm{a}}\right\}$ and $\en-\min\left\{n B_{\mathrm{N}}\big(0,\bar{c}\big),b_{\mathrm{a}}\right\}$ are delivered using drug and radiation administration, respectively.

\subsection{CRT with radio-sensitization mechanism} \label{sec: synergistic}
Chemotherapeutic agents such as hypoxic-cell radio-sensitizers are primarily used to enhance radiation-induced damage and do not exhibit independent cytotoxic activity. For such agents, we assume $\ca=0$ for the target and normal tissue  in (\ref{def: BED}). Substituting $\cat,\can = 0$ in (\ref{eqn: obj}), the objective function of (M) is
\begin{align} \label{eqn: obj_synergistic}
	 \sum_{i=1}^n B_{\mathrm{T}}\left(d_i,c_i\right) = \frac{1}{\overline{\delta^2}}\frac{\abn}{\abt}\en + \left(1-\frac{1}{\sfe}\frac{\abn}{\abt} \right)\sum_{i=1}^n d_i + \left(\frac{\mat}{\man}-\frac{1}{\sfe}\frac{\abn}{\abt} \right) \man \sum_{i=1}^n c_id_i.
\end{align}
To derive the optimal solution to (M) for the radio-sensitization mechanism, we first show the existence of a stationary fractionation regimen.
\begin{proposition} \label{prop: equalFrac}
   If  $\cat,\can = 0$, then there exists an optimal stationary fractionation regimen to (M).
\end{proposition}
The proof is in Appendix \ref{sec: appendix_stationary}. Proposition \ref{prop: equalFrac} states that there exists an optimal regimen $\big(d^*_i,c^*_i\big)\,\left(i=1,\ldots,n\right)$ to (M) in which for all pairs of fractions $i,i^{\prime}$ such that $d^*_i, d^*_{i^{\prime}} > 0$, we have $d^*_i = d^*_{i^{\prime}}$ and $c^*_i = c^*_{i^{\prime}}$. In order to characterize this stationary regimen, we set $c_1=c_2=\ldots=c_n=c$ in (M). Using Lemma \ref{lemma: binding} we see that the optimal total radiation dose, which we denote by $D_{\mathrm{s}}=\sum_{i=1}^n d_i$, as a function of the drug concentration $c$ for a fixed number of treatment fractions $n$ is
\begin{align} \label{eqn: D_synergistic}
D_{\mathrm{s}}\left(c;n\right) = \frac{-\left(1+\man c\right) + \sqrt{\left(1+\man c\right)^2+4\sfe \en / n\bar{\delta}\abn}}{2\sfe/n \abn}. 
\end{align}
This leads to the following result on the optimal fractionation regimen to (M) for CRT with only the radio-sensitization mechanism:

\begin{theorem}\label{theorem: synergistic}
Let
\begin{align*}
	  \Delta_{\mathrm{r}} &= 1-\frac{1}{\sfe}\frac{\abn}{\abt},\\
	\Delta_{\mathrm{s}} &= \frac{\mat}{\man}-\frac{1}{\sfe}\frac{\abn}{\abt},\\
	\varrho_{\mathrm{s}} &= \frac{\man\bar{c}}{ D_{\mathrm{s}}\left(0;1\right) / D_{\mathrm{s}}\left(\bar{c};1\right) - 1},\\
	\underline{\rho}_{\mathrm{s}} &= \frac{2\sfe}{n \abn} D_{\mathrm{s}}\left(\bar{c};n\right)+1, \\
	\overline{\rho}_{\mathrm{s}} &= \frac{2\sfe}{n \abn} D_{\mathrm{s}}\left(0;n\right)+1, \mbox{ and }\\
	c_{\mathrm{s}} & = \frac{\overline{\rho}_{\mathrm{s}}^2-\left(\Delta_{\mathrm{r}}/\Delta_{\mathrm{s}}\right)^2}{ 2 \man \left(\Delta_{\mathrm{r}} / \Delta_{\mathrm{s}} - 1\right)}.
\end{align*}
In the presence of only the radio-sensitization effect, that is, $\cat=\can =0$, the optimal regimen to (M) is
\begin{align*}
\mbox{optimal regimen}:
\begin{cases}
\begin{array}{lll}
	 \mbox{RT-hypo} & c^*=0 &  \Delta_{\mathrm{r}}\leq 0,\Delta_{\mathrm{s}} \leq 0,  \Delta_{\mathrm{r}}/\Delta_{\mathrm{s}} \leq \varrho_{\mathrm{s}} \\
	 \mbox{RT-hypo} & c^*=\bar{c} &  \Delta_{\mathrm{r}}\leq 0,\Delta_{\mathrm{s}} \leq 0,  \Delta_{\mathrm{r}}/\Delta_{\mathrm{s}} > \varrho_{\mathrm{s}}\\
	 \mbox{RT-hypo} & c^*=\bar{c}& \Delta_{\mathrm{r}} \leq 0, \Delta_{\mathrm{s}} \geq 0,  \Delta_{\mathrm{r}}/\Delta_{\mathrm{s}} \leq -\man \bar{c} \\
	 \mbox{RT-std} & c^*=\bar{c}& \Delta_{\mathrm{r}} \leq 0, \Delta_{\mathrm{s}} \geq 0,  \Delta_{\mathrm{r}}/\Delta_{\mathrm{s}} > -\man \bar{c}\\
	\mbox{RT-std}  & c^*=\bar{c} & \Delta_{\mathrm{r}} \geq 0,\Delta_{\mathrm{s}} \geq 0, \Delta_{\mathrm{r}} / \Delta_{\mathrm{s}} < \underline{\rho}_{\mathrm{s}} \\
	\mbox{RT-std} & 0 < c^* = c_{\mathrm{s}} < \bar{c} &  \Delta_{\mathrm{r}} \geq 0,\Delta_{\mathrm{s}} \geq 0, \underline{\rho}_{\mathrm{s}} \leq \Delta_{\mathrm{r}}/\Delta_{\mathrm{s}} < \overline{\rho}_{\mathrm{s}} \\
	\mbox{RT-std} & c^*=0  &   \Delta_{\mathrm{r}} \geq 0,\Delta_{\mathrm{s}} \geq 0,\Delta_{\mathrm{r}} / \Delta_{\mathrm{s}} \geq \overline{\rho}_{\mathrm{s}} \\
	 \mbox{RT-std } & c^*=0 &  \Delta_{\mathrm{r}} \geq 0, \Delta_{\mathrm{s}} \leq 0.
\end{array}
\end{cases}
\end{align*}
\end{theorem}
The proof is in Appendix \ref{sec: appendix_synergistic}. Figure \ref{fig: synergistic} illustrates the results presented in Theorem \ref{theorem: synergistic} for different values of $\Delta_{\mathrm{r}}$ and $\Delta_{\mathrm{s}}$. We make the following observations regarding the impact of radio-sensitizers on the optimal fractionation regimens:
 \begin{itemize}
\item[1.]  Since it is assumed that there is not any independent cytotoxic activity associated with the chemotherapeutic agent, radiation is always administered in the optimal treatment regimen.
\item[2.] The optimal radiation fractionation regimen and drug concentration always depend on the ratio $\Delta_{\mathrm{r}}/\Delta_{\mathrm{s}}$ except for the case of $ \Delta_{\mathrm{r}} \geq 0, \Delta_{\mathrm{s}} \leq 0 $. The transitioning thresholds are specified by $\underline{\rho}_{\mathrm{s}}$, $\overline{\rho}_{\mathrm{s}}$, $\varrho_{\mathrm{s}}$, and $ -\man \bar{c}$.
\item[3.] For standard fractionation regimens, there is a range of $ \Delta_{\mathrm{r}}/\Delta_{\mathrm{s}}$, that is, $\left(\underline{\rho}_{\mathrm{s}},\overline{\rho}_{\mathrm{s}}\right)$, within which an intermediate drug-concentration level is optimal. However, for hypo-fractionated regimens, the drug is either not used at all or delivered at maximum concentration level.
\item[4.] The administration of radio-sensitizers may alter the fractionation decision in radiotherapy. More specifically, there are scenarios ($ \Delta_{\mathrm{r}} \leq 0, \Delta_{\mathrm{s}} \geq 0 $) in which a hypo-fractionated regimen is optimal if no radio-sensitizer is used; however, in the presence of the radio-sensitizer, it is optimal to switch to a standard fractionation regimen. This aspect will be further discussed in Section \ref{sec: discussion}.
\end{itemize}

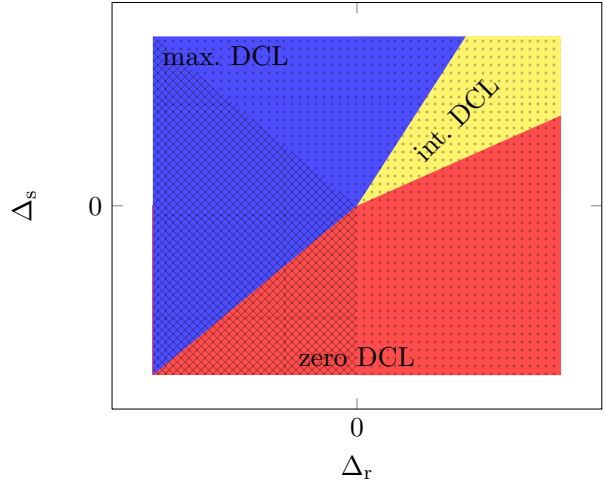
\begin{figure}[h]
\center
\begin{tikzpicture}[scale=0.95]
	\begin{axis}[xtickmin=0,xtickmax=0,ytickmin=0,ytickmax=0, xlabel=$\Delta_{\mathrm{r}}$,ylabel=$\Delta_{\mathrm{s}}$,domain=-15:15,
			area legend, legend style={anchor=north west}]
		\addplot [fill=yellow!70,draw=none] coordinates {(0,0) (0,15) (15,15) (15,0)} \closedcycle;
		\addplot [fill=red!70,draw=none] coordinates {(0,0) (0,-15) (-15,-15) (-15,0)} \closedcycle;
		\addplot [fill=red!70,draw=none] coordinates {(0,0) (0,-15) (15,-15) (15,8) (0,0)} \closedcycle;
		\addplot [fill=blue!70,draw=none] coordinates {(0,0) (-15,-15) (-15,15) (8,15) (0,0)} \closedcycle;
		\addplot [ fill=gray!80,opacity = 0.5,pattern=dots,draw=none] coordinates {(0,-15) (0,0) (-15,15) (0,15) (15,15) (15,-15) (0,-15)} \closedcycle;
		\addplot [ fill=gray!80,opacity = 0.5,pattern=crosshatch,draw=none] coordinates {(0,-15) (0,0) (-15,15) (-15,-15) (0,-15) } \closedcycle;
		\node at (axis cs: -15,15) [anchor=north west] {max.\ DCL};
		\node at (axis cs: 0,-15) [anchor=south]{zero DCL};
		\node at (axis cs:   3,5) [anchor=north west,rotate=45]{int.\ DCL};	
    \end{axis}
\end{tikzpicture}
\caption{\label{fig: synergistic} \small{Schematic of optimal regimens for CRT with only the radio-sensitization mechanism. Zero, intermediate, and maximum drug-concentration level (DCL) correspond to $c^*=0$, $0<c^*< \bar{c}$, and $c^*=\bar{c}$, respectively. Dotted and hatched regions illustrate whether standard or hypo-fractionated radiotherapy is optimal, respectively. }}
\end{figure}

\subsection{CRT with additivity and radio-sensitization mechanisms} \label{sec: combined}
Several chemotherapeutic agents such as Gemcitabine and Cisplatin are suggested to have radio-sensitization properties as well as independent cytotoxic activity \citep{Hall2006}. For those agents that exhibit both additivity and radio-sensitization mechanisms, using the definition of $\Delta_{\mathrm{r}}$, $\Delta_{\mathrm{a}}$, and $\Delta_{\mathrm{s}}$, the objective function of (M) in (\ref{eqn: obj}) is
\begin{align*}
	 \sum_{i=1}^n B_{\mathrm{T}}\left(d_i,c_i\right) &= \frac{1}{\overline{\delta^2}}\frac{\abn}{\abt}\en + \Delta_{\mathrm{r}} \sum_{i=1}^n d_i  + \Delta_{\mathrm{a}} \left(\can \sum_{i=1}^n c_i\right) +  \Delta_{\mathrm{s}}\left( \man \sum_{i=1}^n c_id_i \right).
\end{align*}
To obtain optimal treatment regimens for this case, we start with the following special cases:
\begin{itemize}
	\item[1.] $\Delta_{\mathrm{r}} \geq 0$, $\Delta_{\mathrm{a}} < 0$, and $\Delta_{\mathrm{s}} < 0$: it is easy to see that the optimal drug concentration is $c_i^*=0\;(i=1,\ldots,n)$ since $\Delta_{\mathrm{a}} < 0$ and $\Delta_{\mathrm{s}} < 0$  discourage any positive drug concentration. Therefore, the problem reduces to radiotherapy alone. In particular, since $\Delta_{\mathrm{r}} \geq 0$, a standard radiation fractionation scheme is optimal. 
	\item[2.] $\Delta_{\mathrm{r}} < 0$, $\Delta_{\mathrm{a}} \geq 0$, and $\Delta_{\mathrm{s}} < 0$: in this case, $\Delta_{\mathrm{r}} < 0$ and $\Delta_{\mathrm{s}} < 0$ discourage any positive radiation dose. Thus, the optimal treatment regimen uses only chemotherapy, that is, $d_i^*=0 \;(i=1,\ldots,n)$. However, this is only the case if condition (\ref{condition: add}) is satisfied; otherwise, $d_i^*>0$ for some $i$ and thus the optimal regimen also involves radiotherapy.
\end{itemize}
In general, it may not be possible to derive closed-form expressions for the optimal treatment regimens if both mechanisms are present. This is mainly because optimal stationary regimens may not necessarily exist. In fact, the numerical example in Section \ref{sec: example} shows that chemotherapeutic agents with both mechanisms may give rise to optimal non-stationary regimens. Therefore, we develop a DP algorithm to obtain optimal fractionation regimens to (M). The \emph{stages} of the DP algorithm refer to treatment fractions $i=1,\ldots,n$. The \emph{control variables} are the radiation dose and drug concentration $\left(d_i,c_i\right)$ used at stage $i\,\left(i=1,2,\ldots,n\right)$. The \emph{state} of the system at the beginning of stage $i$, denoted by $x_i$, is the mean of the cumulative BED in the normal tissue. Thus, the state dynamics is 
\begin{align*}
	 x_{i+1} &= x_i + \frac{1}{|V|}\sum_{v\in V} B_{\mathrm{N}}\left(\delta_vd_{i},c_{i}\right) & i & =1,\ldots,n,
\end{align*}
where $x_{n+1}$ represents the state of the system at the end of the treatment. At each stage $i\,\left(i=1,\ldots,n\right)$ we consider an intermediate reward of $B_{\mathrm{T}}\left(d_i,c_i \right)$ measuring the BED delivered to the target during that stage. To ensure that the solution satisfies the constraint in (\ref{eqn: BED}), we also define a terminal reward function, denoted by $R_{n+1}\left(x_{n+1}\right)$, as follows: 
\begin{align*}
	R_{n+1}\left(x_{n+1}\right) &= \begin{cases}
	0 & x_{n+1}\leq \en ;\\
	-\infty & x_{n+1} > \en. \end{cases}
\end{align*} 
Associated with stage $i\,\left(i=1,\ldots,n\right)$ we define the \emph{cost-to-go} function, denoted by $J_i \left(x_i \right)$, as the maximum cumulative reward that can be attained in remaining treatment fractions $i,\ldots,n$ when starting with state $x_i$. We can then write the Bellman's recursive equations as follows:
\begin{align}
	J_i \left(x_i\right) &= \max_{\substack{0 \leq c_i \leq \bar{c} \\ d_i \geq 0}} \left\{ B_{\mathrm{T}}\left(d_i,c_i \right) +  J_{i+1} \left( x_i + \frac{1}{|V|}\sum_{v\in V} B_{\mathrm{N}}\left(\delta_vd_{i},c_{i}\right) \right) \right\} & i & =1,\ldots,n  \label{eqn: bellman1} \\
	J_{n+1} \left(x_{n+1}\right) &= R_{n+1}\left(x_{n+1}\right). \label{eqn: bellman2} 
\end{align}
The Bellman's equations can be solved using backward recursion. This requires a discretization of the state variables $x_i\,\left(i=1,\ldots,n\right)$. In oder to increase the accuracy of the evaluation of $J_i$ at non-discretized $x_i$ values, a linear interpolation of appropriate discretized values is performed. Our DP framework is used to find optimal solutions to (M) when both additivity and radio-sensitization mechanisms are relevant.

\section{Numerical examples}\label{sec: example}
In this section we discuss how the results obtained in Section \ref{sec: solution} can be applied to treatment plans designed for individual patients. As a proof of principle, we use a photon and proton treatment plan for a locally-advanced lung cancer case, which are extracted from \cite{Zhang2010}. In particular, we assume that the target receives a uniform dose while the normal lung, as the dose-limiting normal tissue, receives a heterogeneous dose distribution. \emph{Dose-volume histogram} (DVH) curves for the normal lung, associated with the photon and proton plans, are illustrated in Figure \ref{fig: DVH}. Clearly, the proton plan allows for a better sparing of the normal lung. For each DVH curve, the first and second moments of the sparing factors $\delta_v \,(v \in V)$, which are input parameters to the mathematical results in Section \ref{sec: solution}, are calculated. We then use the results in Section \ref{sec: solution} to determine the optimal fractionation regimens and drug administration schedules for these two radiotherapy plans. We use common BED parameters of $\abt=10$ Gy and $\abn=4$ Gy for the target and normal lung, respectively. We let $\en=25$ Gy as the upper bound on the average normal-lung BED. We also set $\bar{c}=1$ considering a normalized maximum drug concentration. Moreover, we consider a maximum number of $n=30$ treatment fractions.

\begin{figure}[h]
\center
\begin{tikzpicture}[scale=0.8]
\pgfplotsset{height=6cm, width=9cm, no markers}
\begin{axis}[scale only axis, xmin=0, xmax=1, ymin=0, ymax=1,xtick={0,0.2,0.4,0.6,0.8,1}, ylabel={fractional volume}, xlabel={fractional dose},grid=major]
		\addplot[thick,color=black] coordinates {
( 0,0.997)
(0.0131,0.934)
(0.0415,0.798)
(0.0459,0.789)
(0.0852,0.708)
(0.0873,0.705)
(0.133,0.649)
(0.175,0.596)
(0.177,0.595)
(0.216,0.567)
(0.221,0.564)
( 0.26,0.541)
(0.264,0.539)
(0.306,0.525)
(0.308,0.524)
(0.352,0.504)
(0.393,0.482)
(0.437,0.459)
( 0.48,0.435)
(0.522,0.406)
(0.524,0.404)
(0.566,0.375)
(0.568,0.374)
(0.609,0.345)
(0.655,0.321)
(0.657,0.319)
(0.697,0.289)
( 0.74,0.271)
(0.745,0.269)
(0.784,0.253)
(0.786,0.252)
(0.828,0.236)
(0.869,0.207)
(0.915, 0.17)
(0.937,0.127)
(0.939,0.122)
( 0.95,0.0797)
(0.956,0.0664)
(0.969,0.0398)
(1,0)};
\addplot[thick,color=black,dashed] coordinates {
(0,0.995)
(0.0131, 0.54)
(0.0415,0.505)
(0.0459,  0.5)
(0.0852,0.473)
(0.0873,0.471)
(0.133, 0.45)
(0.175,0.435)
(0.177,0.434)
(0.216,0.422)
(0.221,0.421)
( 0.26,0.407)
(0.264,0.405)
(0.306,0.395)
(0.308,0.394)
(0.352,0.379)
(0.393,0.368)
(0.437, 0.36)
( 0.48,0.342)
(0.522,0.324)
(0.524,0.323)
(0.566,0.307)
(0.568,0.306)
(0.609,0.289)
(0.655,0.271)
(0.657, 0.27)
(0.697,0.252)
( 0.74,0.237)
(0.745,0.236)
(0.784,0.218)
(0.786,0.217)
(0.828,0.199)
(0.869,0.173)
(0.915,0.143)
(0.937,0.0956)
(0.939,0.0909)
( 0.95,0.0672)
(0.956,0.0531)
(0.969,0.0372)
(1,0)
};
\node at (axis cs:0.4,0.6) [anchor=north west] {Photon plan};
\node at (axis cs: 0.3,0.25) [anchor=south]{Proton plan};
\end{axis}
\end{tikzpicture}
\caption{\label{fig: DVH} \small{DVH curves of the fractional dose in the normal lung for the photon (solid) and proton (dashed) treatment plans provided in \cite{Zhang2010}, for a locally-advanced lung cancer case. }}
\end{figure}
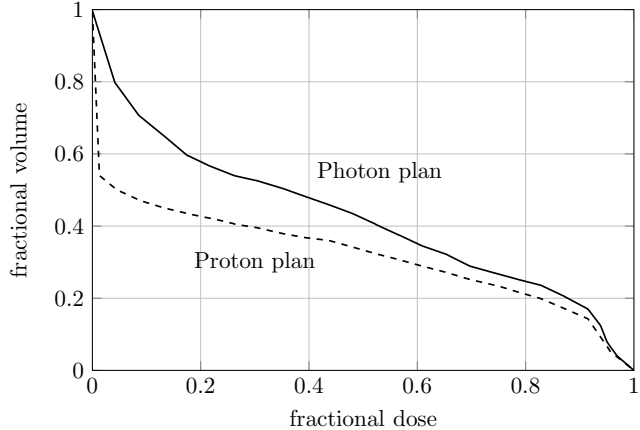

\paragraph{Radiotherapy} Using equations from \cite{Unkelbach2013} to determine the first and second moments of sparing factors from the DVH curves, we obtain $\bar{\delta}=0.42$, $\overline{\delta^2}=0.31$ and $\sfe = 0.74$ for the photon plan and $\bar{\delta}=0.32$, $\overline{\delta^2}=0.25$, and $\sfe=0.78$ for the proton plan. We first consider the case of radiotherapy alone and calculate $\Delta_{\mathrm{r}}$, which is 0.46 and 0.49 for the photon and proton plan, respectively. Since $\Delta_{\mathrm{r}} \geq 0$, a stationary standard fractionation regimen is optimal for both the photon and proton plans.

\paragraph{CRT with only the  additivity mechanism}  According to Theorem \ref{theorem: additive}, since $\Delta_{\mathrm{r}} \geq 0$, in order to use the drug in the optimal regimen, condition $\Delta_{\mathrm{r}}/\Delta_{\mathrm{a}} \leq \overline{\rho}_{\mathrm{a}}$ should be satisfied, which yields $\Delta_{\mathrm{a}} \geq  0.70$ and $\Delta_{\mathrm{a}} \geq  0.88$ for the photon and proton plan, respectively. Therefore, the optimal regimen uses the drug along with radiation if $\cat/\can \geq 1.99$ and $\cat/\can \geq 2.48$ for the photon and proton plan, respectively. This suggests that, to use the chemotherapeutic agent with the proton plan, the relative additive effect of the drug in the target versus normal lung ($\cat/\can$) should be 25\% larger compared to the photon plan. The difference between the threshold values is attributed to the better sparing of the normal lung achieved using the proton plan, rendering the use of a systemic agent with a small relative additive effect unattractive. Furthermore, the optimal treatment regimen uses only the chemotherapeutic agent if $\cat/\can \geq 2.39$ and $\cat/\can \geq 3.13$  in case of the photon and proton plan, respectively. Thus, it is optimal to combine the photon and proton plan with chemotherapeutic agents only if the corresponding relative additive effect lies within $\left[1.99,2.39 \right]$ and $\left[2.48,3.13\right]$, respectively.

\paragraph{CRT with only the radio-sensitization mechanism} We next investigate the impact of radio-sensitizers on the optimal regimens for the two plans. According to Theorem \ref{theorem: synergistic}, since $\Delta_{\mathrm{r}} \geq  0$, in order to have a positive drug concentration in the optimal regimen, condition $\Delta_{\mathrm{r}}/\Delta_{\mathrm{s}} \leq \overline{\rho}_{\mathrm{s}}$ should be satisfied, which yields $\Delta_{\mathrm{s}} \geq  0.29$ and $\Delta_{\mathrm{s}} \geq  0.28$, or equivalently  $\mat/\man \geq 0.83$ and $\mat/\man \geq 0.79$ for the photon and proton plan, respectively. Hence, the relative radio-sensitization effect in the target versus normal lung ($\mat/\man$) for the photon plan is required to be 5\% larger compared to the proton plan. The difference in the threshold values is due to the larger radiation dose delivered by the photon plan in the normal lung. This, in turn, increases the radio-sensitization activity and thus the radiation toxicity in the normal lung. In contrast, a proton plan can achieve a better sparing of the normal lung, thereby reducing the radio-sensitization activity. 

\paragraph{CRT with additivity and radio-sensitization mechanisms} Finally, we consider the impact of chemotherapeutic agents with combined mechanisms on the optimal regimens. We apply our DP algorithm discussed in Section \ref{sec: solution} to the photon and proton treatment plans for different values of $\cat, \can, \mat, \man \geq 0$. More specifically, we solve instances of (M) where $\cat$, $\can$, $\mat$, and $\man$ assume values in $\left\{0.1,0.15,\ldots,1\right\}$ corresponding to scenarios in which the relative importance of additivity and radio-sensitization effects in (\ref{def: BED}) varies from 10\% to 100\% of the radiation BED. A discretization step of 0.05 was used for the state and control variables. An important observation is that the inclusion of both effects may lead to non-stationary optimal regimens. Figure \ref{fig: non_stationary} depicts an example of a non-stationary regimen for the photon plan with $\Delta_{\mathrm{a}} < 0$ and $\Delta_{\mathrm{s}} > 0$ indicating that the additive and radio-sensitization effect is unfavorable and favorable, respectively. The optimal regimen uses the chemotherapeutic agent in only a subset of treatment fractions with an escalated radiation dose in those fractions to benefit from the radio-sensitization mechanism. Note that the model does not consider the sequence of fractions used in the fractionation regimen, and thus the solutions are insensitive to the order according to which the fractions are administered. This case is further discussed in Section \ref{sec: discussion}. Figure \ref{fig: combined} illustrates the optimal regimens for different parameter values obtained by the DP algorithm for the photon treatment plan. One major difference of the schematic of optimal regimens for the combined mechanisms from those of individual mechanisms is the lack of sharp boundaries between sub regions. More specifically, except for the third quadrant ($\Delta_{\mathrm{a}} < 0$, $\Delta_{\mathrm{s}} < 0$), where it can be analytically shown that RT-std is the only optimal treatment regimen (see Section \ref{sec: combined}), the boundaries are not exclusive. Therefore, the dominant treatment regimen within each sub region is illustrated in Figure \ref{fig: combined}. In particular, in the first quadrant ($\Delta_{\mathrm{a}} \geq 0$, $\Delta_{\mathrm{s}} \geq 0$), where both additive and radio-sensitization effects are favorable, stationary CRT regimens are the dominant form of optimal treatment regimens. In the second quadrant ($\Delta_{\mathrm{a}} < 0$, $\Delta_{\mathrm{s}} \geq 0$), where the additive effect is unfavorable, non-stationary CRT regimens are optimal. The example in Figure \ref{fig: non_stationary} belongs to this quadrant. The optimal treatment regimen in the third quadrant ($\Delta_{\mathrm{a}} < 0$, $\Delta_{\mathrm{s}} < 0$) is radiotherapy only using a standard fractionation regimen, as discussed in Section \ref{sec: combined}. Finally, in the fourth quadrant ($\Delta_{\mathrm{a}} \geq 0$, $\Delta_{\mathrm{s}} < 0$), where the radio-sensitization effect is unfavorable, it might be still desirable to administer the chemotherapeutic agent sequentially with radiation (rather than concurrently) to avoid the unfavorable radio-sensitization effect and at the same time to benefit from the additive effect. This leads to the dominance of sequential CRT (neoadjuvant or adjuvant) regimens in this quadrant.

\begin{figure}[h]
\center
\begin{tikzpicture}[scale=0.8]
\pgfplotsset{height=5cm, width=11cm,no markers}
\begin{axis}[scale only axis, xmin=0, xmax=60, xticklabels={1,5,10,15,20,25,30}, xtick={1,9,19,29,39,49,59}, ymin=0, ymax=5, bar width = 4.5pt, ylabel={radiation (Gy)}, ylabel style = {red}, yticklabel style={color=red}, xlabel={fraction \#}]
	\addplot[ybar,color=red, fill] coordinates {(1,4.3) (3,0.9) (5,0.9) (7,0.9) (9,0.9) (11,0.9) (13,0.9) (15,0.9) (17,0.9) (19,4.3) (21,0.9) (23,0.9) (25,0.9) (27,0.9) (29,0.9) (31,0.9) (33,0.9) (35,0.9) (37,0.9) (39,4.3) (41,0.9) (43,0.9) (45,0.9) (47,0.9) (49,0.9) (51,0.9) (53,0.9) (55,0.9) (57,0.9) (59,0.9)};
\end{axis}

\begin{axis}[scale only axis, axis y line*=right, xtick={1,9,19,29,39,49,59}, xticklabels= {}, ytick= {0,1,2,3,4,5},yticklabels={0,$\frac{1}{3}$,$\frac{2}{3}$,1,$\frac{4}{3}$,$\frac{5}{3}$},yticklabel style={color=blue}, ylabel={ DCL}, ylabel style={blue,at={(1.2,0.5)}}, xmin=0, xmax=60, ymin=0, ymax=5]
	\addplot[bar width = 4.5pt, ybar,color=blue, fill] coordinates {(2,3) (20,3) (40,3)};
    \end{axis}	
\end{tikzpicture}
\caption{\label{fig: non_stationary} \small{An example of non-stationary optimal treatment regimens (insensitive to the fraction sequence) for the photon plan where $\Delta_{\mathrm{r}} > 0$, $\Delta_{\mathrm{a}} < 0$, and $\Delta_{\mathrm{s}} > 0$, using $\cat = 0.3$, $\can = 0.65$, $\mat = 0.65$, and $\man= 0.25$ (DCL: drug-concentration level).}}
\end{figure}
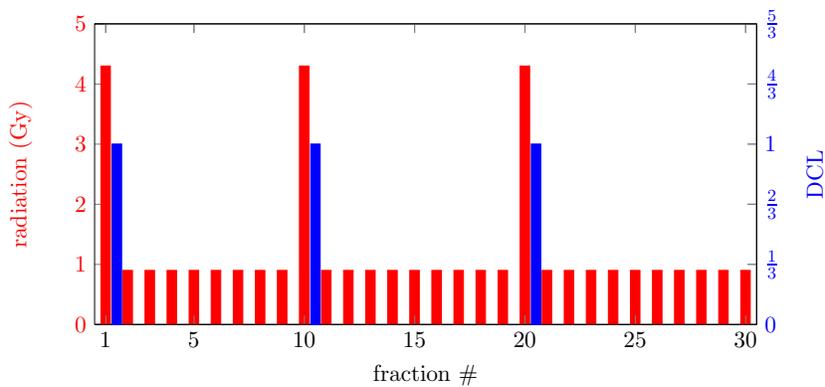

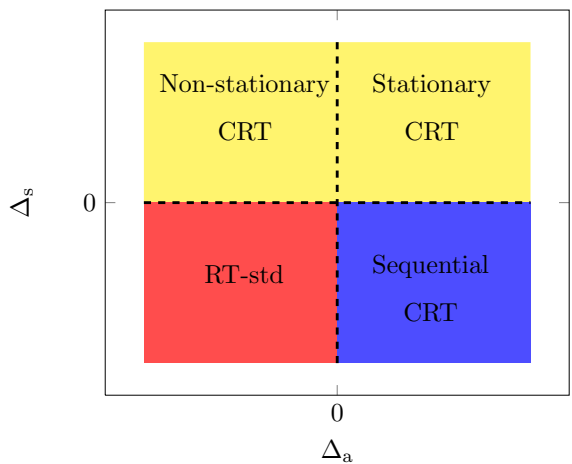
\begin{figure}[h]
\center
\begin{tikzpicture}[scale=0.90]
	\begin{axis}[xtickmin=0,xtickmax=0,ytickmin=0,ytickmax=0, xlabel=$\Delta_{\mathrm{a}}$,ylabel=$\Delta_{\mathrm{s}}$,domain=-15:15,area legend, legend style={anchor=north west}]
		\addplot [fill=yellow!70,draw=none] coordinates {(-15,0) (-15,15) (15,15) (15,0) (0,0) (-15,0)} \closedcycle;
		\addplot [fill=red!70,draw=none] coordinates {(0,0) (-15,0) (-15,-15) (0,-15) (0,0)} \closedcycle;
		\addplot [fill=blue!70,draw=none] coordinates {(0,0) (15,0) (15,-15) (0,-15) (0,0)} \closedcycle;
		\node at (axis cs: -14.5,5) [anchor=south west,align=center] {Non-stationary \\ CRT};
		\node at (axis cs: 2,5) [anchor=south west,align=center] {Stationary \\ CRT};
		\node at (axis cs: -11,-5) [anchor=north west]{RT-std};
		\node at (axis cs:  2,-4) [anchor=north west,align=center]{Sequential \\  CRT};
		\draw[very thick,dashed] (axis cs: 0,-15) -- (axis cs: 0,15);
		\draw[very thick,dashed] (axis cs: -15,0) -- (axis cs: 15,0);
    \end{axis}
\end{tikzpicture}
\caption{\label{fig: combined} \small{Schematic of optimal regimens for CRT with additivity and radio-sensitization mechanisms obtained by the DP algorithm for the photon plan with $\Delta_{\mathrm{r}} = 0.46$.}}
\end{figure}


\section{Discussion} \label{sec: discussion}
In this section we comment on the CRT fractionation regimens provided by the model. In particular, using the results discussed in Sections \ref{sec: solution} and \ref{sec: example}, we make the following observations:  
\paragraph{Diminishing return in target BED leads to combined-modality treatments}
For the case of radiotherapy alone, using equations (\ref{eqn: obj_rad}) and (\ref{eqn: D_rad}) and substituting $\en$ with variable $b$, we can write the target BED as a function of average normal-tissue BED for a given number of treatment fractions $n$ as follows:
\begin{align*}
	 H\left(b;n \right) = \frac{1}{\overline{\delta^2}}\frac{\abn}{\abt}b + \Delta_{\mathrm{r}} \frac{-1 + \sqrt{1+4\sfe b / n\bar{\delta}\abn}}{2\sfe/n\abn}.
\end{align*}
In particular, $H\left(\cdot;n\right)$ characterizes the trade-off between the target BED and the average normal-tissue BED. It is easy to verify that $H\left(\cdot;n\right)$ is a concave function if $\Delta_{\mathrm{r}} \geq 0$, and convex otherwise (see Appendix \ref{sec: appendix_tradeoff}). Therefore, when standard fractionation regimen is optimal, the concavity of the trade-off curve suggests a diminishing return in the target BED as one allows for an increase in the average normal-tissue BED. This implies that the introduction of a second modality such as chemotherapeutic agents with additivity mechanism, may lead to a larger gain in the target BED rather than solely escalating the radiation dose. Thus, the diminishing return in the target BED explains the existence of CRT optimal regimens in Theorem \ref{theorem: additive} despite assuming no interactive cooperation between the drug and radiation.


\paragraph{Chemotherapeutic agents with additivity mechanism do not change optimal radiation fractionation regimens} 
According to Theorem \ref{theorem: additive}, the condition for optimality of  hypo- and standard fractionation regimen is $\Delta_{\mathrm{r}} < 0$ and $\Delta_{\mathrm{r}} \geq 0$, respectively. These conditions are identical to those of the radiotherapy-alone case, which suggests that the introduction of chemotherapeutic agents with only the additivity mechanism does not lead to a change in the optimal radiation fractionation scheme. However, the relative additive effect in the target versus the normal tissue ($\cat/\can$) determines the extent to which the drug is used in the optimal regimen. In particular, Figure \ref{fig: additive} illustrates that if $\Delta_{\mathrm{r}} \geq 0$ and $\Delta_{\mathrm{a}} \geq 0$ (first quadrant), then as $\cat/\can$ and, in turn, $\Delta_{\mathrm{a}}$ increases, the optimal regimen transitions from RT-std to CRT-std and from CRT-std to CT. Similarly, if $\Delta_{\mathrm{r}} < 0$ and $\Delta_{\mathrm{a}} < 0$ (third quadrant), then as $\cat/\can$ and thus $\Delta_{\mathrm{a}}$ increases, the optimal regimen changes from RT-hypo to CT.

\paragraph{Radio-sensitizers may change optimal hypo-fractionated radiation regimens} 
In contrast to the additivity mechanism, the radio-sensitization mechanism may lead to a change in the optimal radiation fractionation scheme. More specifically, if $\Delta_{\mathrm{r}} < 0$, then in the absence of radio-sensitizers, a hypo-fractionation radiation regimen is optimal. However, this is not necessarily the case if a radio-sensitizer is being introduced. Figure \ref{fig: synergistic} shows that if $\Delta_{\mathrm{r}} < 0$ and $\Delta_{\mathrm{s}} \geq 0$ (second quadrant), then the maximum drug concentration allowed is optimal. Moreover, if the relative radio-sensitization effect in the target versus normal tissue ($\mat/\man$) is sufficiently large, then the optimal radiation scheme changes from hypofrcationation to standard fractionation. 
This change in the fractionation regimen is explained as follows: since the maximum drug concentration is administered at each fraction $i=1,2,\ldots,n$, then (\ref{def: BED}) can be rewritten as
\begin{align*}
  	\frac{B \left(d_i,\bar{c} \right)}{1+\ma \bar{c}} &= d_i \left(1 + \frac{d_i}{\ab \left(1+\ma \bar{c}\right)}\right).
\end{align*}
This suggests the use of condition (\ref{condition: RT}) as discussed in Section \ref{sec: radiation} to determine the optimal radiation fractionation regimen using an \emph{effective} $\ab$ ratio for the target and normal tissue, which are $\abt \left(1+\mat \bar{c}\right)$ and $\abn \left(1+\man \bar{c}\right)$, respectively. Thus, for the case of $\Delta_{\mathrm{r}} < 0$ and $\Delta_{\mathrm{s}} \geq 0$, if
\begin{align*}
	\abn\left(1+\man \bar{c}\right) \leq \sfe \abt\left(1+\mat \bar{c}\right),
\end{align*}
then a standard fractionation regimen is optimal; otherwise, a hypo-fractionation regimen is optimal. In particular, if $\mat$ is sufficiently larger than $\man$ (case of chemotherapeutic agents with preferential radio-sensitization activity in the target versus normal tissue), then the effective $\ab$ ratio in the target is larger than that of the normal tissue, thereby tipping the scales toward standard radiation fractionation regimen.
  
\paragraph{Concurrent CRT with both mechanisms may favor non-stationary fractionation regimens} 
The numerical example illustrated in Figure \ref{fig: non_stationary} presents a case for which a non-stationary treatment regimen is optimal. To provide an explanation for the existence of optimal non-stationary regimens, we consider the additivity and radio-sensitization mechanisms in this example independently: for the additivity mechanism, since $\Delta_{\mathrm{r}} \geq 0$ and $\Delta_{\mathrm{a}} < 0$, Theorem \ref{theorem: additive} does not recommend the administration of the chemotherapeutic agent. On the other hand, for the radio-sensitization mechanism, since $\Delta_{\mathrm{r}} \geq 0$ and $\Delta_{\mathrm{s}} \geq 0$, Theorem \ref{theorem: synergistic} shows that the use of a standard radiation fractionation regimen along with a maximum drug concentration is optimal. Hence, if both mechanisms are simultaneously considered, then the optimal regimen uses the drug in only a subset of fractions with an escalated radiation dose in those fractions. This example may represent the scenario in which the chemotherapeutic agent has relatively high toxicity in the normal tissue (small $\cat/\can$) and at the same time relatively strong radio-sensitization effect in the target (large $\mat/\man$). Therefore, for such chemotherapeutic agents, the use of non-stationary treatment regimens may be optimal.

\paragraph{Decision on administration of chemotherapeutic agents depends on the spatial dose distribution} 
The numerical example in Section \ref{sec: example} compares the fractionation regimens for a photon and proton treatment plan. The difference between the two is attributed to a better sparing of the normal lung achieved by the proton therapy and thus smaller values of sparing factors, $\bar{\delta}$ and $\overline{\delta^2}$ (note that smaller sparing factors do not necessarily lead to a smaller $\sfe$). In particular, in order to use chemotherapeutic agents with only the additivity mechanism along with the proton treatment, larger values of $\cat/\can$ are required compared to the photon treatment. This is because administration of a systemic agent with the proton treatment is not beneficial unless its relative additive effect ($\cat/\can$) is sufficiently large. In other words, chemotherapeutic agents that are used concomitantly with proton plans are expected to have larger selectivity in their additivity mechanism. On the other hand, in order to use radio-sensitizers along with proton treatments, smaller values of $\mat/\man$ are acceptable compared to photon treatments. This is because photon treatments typically deliver low doses of radiation to a large part of the lung, which is so-called \emph{low-dose bath} (see the DVH curve of normal lung for the photon plan in Figure \ref{fig: DVH}). The radio-sensitization mechanism increases the normal-lung toxicity caused by the low-dose bath. In contrast, the proton treatment has a significantly smaller low-dose bath, thereby reducing the radio-sensitization activity in the normal lung. Hence, it may be still beneficial to administer radio-sensitizers with a smaller preferential effect (smaller $\mat/\man$) along with the proton treatment. 

\paragraph{Limitations of the study and future research directions}
Although the mechanism of action for chemotherapeutic agents are suggested to be mainly of the additivity and radio-sensitization forms, the extent of these effects may vary across different chemotherapeutic agents and patients. Hence, the quantification of these effects is challenging \citep{Wheldon1988}. Moreover, the validity of the BED model and its parameter estimate have been questioned (see, e.g., \cite{Brenner2008,Kirkpatrick2008}). Therefore, the role of mathematical optimization in CRT is not to recommend alternative clinical practice but to generate hypotheses that can inspire the design of new clinical trials. This study presents an initial attempt at developing a mathematical framework for hypothesis generation in CRT treatment planning. This framework can be further extended to accommodate additional clinical considerations as follows:
\begin{itemize}
	\item For the proof of concept, we used a linear and multiplicative term in (\ref{def: BED}) to model the additivity and radio-sensitization mechanisms, respectively. As an extension to this work, we plan to study the structure of optimal treatment regimens for generalized BED models where these mechanisms are accounted for using functions of other forms in terms of the drug's concentration and radiation dose. 
	\item The presented results are based on the assumption that the dose-response parameters can be accurately estimated. However, there may be uncertainty associated with parameter estimation or variation across different patients. \cite{Davison2011} studied the impact of the uncertainty in radio-sensitivity parameters on the optimal radiotherapy fractionation regimens and cautioned against the resulting variability in tumor-cell kill. In particular, they characterized optimal fractionation regimens with minimal risk of treatment failure for different uncertainty levels of the LQ model parameters. Furthermore, it is suggested that some tumor reoxygenation may occur after a few fractions, altering the tumor radio-sensitivity during the course of the treatment \citep{carlson2006}. Our framework can be extended to account for these sources of parameter uncertainty. 
	\item A major factor impacting the fractionation decision is tumor repopulation during the course of the treatment, particularly in treatment sites with a fast proliferating rate such as head-and-neck cancers. If unaccounted for, tumor repopulation may severely compromise the treatment outcome (\cite{Withers1988}). \cite{Ramakrishnan2013} has shown that in contrast to \emph{exponential} repopulation for which a stationary and finite radiation fractionation regimen is optimal, the incorporation of \emph{accelerated} repopulation may lead to non-stationery optimal regimens. Although in our current framework the repopulation effect has not been explicitly accounted for, the bound enforced on the maximum number of treatment fractions limits this effect. Our framework can be extended to explicitly incorporate the impact of accelerated repopulation in the fractionation decision for chemoradiotherapy.
	\item It is assumed in our current framework that the drug administered at each treatment fraction does not carry over to subsequent fractions. This assumption has served as the basis for formulating the set of constraints in (\ref{eqn: drug}). However, this may not be necessarily valid for some chemotherapeutic agents. Therefore, our framework should be tailored for individual chemotherapeutic agents by accommodating the specifics of the drug pharmacokinetics.      
	\item Our current formulation considers the CRT fractionation decision for a given spatial dose distribution. More specifically, it is assumed that the spatial aspect of the treatment plan is already determined and only the \emph{temporal} aspect is optimized. Hence, the optimal solution to the CRT fractionation problem may only scale the spatial dose distribution up or down. However, as previously discussed in Section \ref{sec: discussion}, the radiation dose distribution influences the radio-sensitization effect. Therefore, fixing the spatial dose distribution limits the therapeutic gain that might be achieved through optimizing the fractionation regimen. This has  been rightfully pointed out in other studies (see, e.g., \cite{bernier2003}) stating that concurrent CRT may change the desirable radiation dose distribution with respect to the normal tissue and target volume. \emph{Spatio-temporal planning} aims at simultaneously optimizing for the spatial and temporal aspects of the treatment plan and has been previously studied for radiotherapy (see, e.g., \cite{Kim2012}). We plan to investigate the extension of spatio-temporal treatment planning to concurrent CRT.    
\end{itemize}

\section{Conclusion} \label{sec: conclusion}
We developed a mathematical framework to study the impact of chemotherapeutic agents on optimal fractionation regimens. We considered an extension of the BED model to incorporate two major mechanisms of action for chemotherapeutic agents, which are additivity and radio-sensitization. We then derived closed-form solutions to the fractionation problem for each individual mechanism, and developed a DP algorithm to solve the fractionation problem with combined mechanisms. Results suggest that chemotherapeutic agents with only an additive effect do not change the optimal radiation fractionation regimens; however, radio-sensitizers may change hypo-fractionation regimens to standard ones. Moreover, chemotherapeutic agents with both effects may give rise to non-stationary regimens. Lastly, the spatial dose distribution may impact the use of chemotherapeutic agents along with radiation treatments. The demonstrated results motivate future research to extend this framework to spatio-temporal planning and to account for additional clinical considerations related to tumor repopulation as well as the mechanism of action and pharmacokinetics of chemotherapeutic agents. Such framework may assist clinicians with hypothesis generation to design novel CRT fractionation schemes that can be tested in clinical trials.

\section*{Acknowledgements}
The authors would like to thank David Craft, Jagdish Ramakrishnan, and  Edwin Romeijn for helpful discussions and feedback.

\newpage

\newpage
\begin{appendices}
\appendixpage 
\numberwithin{equation}{section}
\setcounter{equation}{0}

\section{Extending the BED Model to Concurrent CRT} \label{sec: appendix_BEDExt}
In this section, we motivate the extended BED model defined in Section \ref{sec: crtfx-dose-response} through an extension of the LQ model to describe in-vitro cell-survival assays in the presence of chemotherapeutic agents. The radiation LQ model is
\[S_{\mathrm{RT}}(d)=e^{-\left(\alpha d + \beta d^2 \right)}, \]
where $S_{\mathrm{RT}}$ denotes the fraction of cells surviving radiation dose $d$ and $\alpha$ and $\beta$ are cell-specific parameters. In addition, in-vitro assay studies suggest an exponential cell-survival curve for the chemotherapeutic agent as follows:
\[S_{\mathrm{CT}}(c) = e^{-\gamma c}, \]
where $S_{\mathrm{CT}}$ denotes the fraction of cells surviving drug concentration $c$ and $\gamma$ is a cell-specific parameter. The unit of the parameter $\gamma$ is chosen such that $\gamma c$ is a dimensionless quantity. A non-interactive cooperation between the chemotherapeutic agent and radiation yields a survival fraction of
\begin{align*}
	 S_{\mathrm{RT}}(d)S_{\mathrm{CT}}(c) = e^{-\left(\alpha d + \beta d^2 +\gamma c\right)}
\end{align*}
if cells are exposed to radiation dose $d$ and drug concentration $c$. Furthermore, in-vitro assay experiments on radio-sensitization indicate an increase in the parameter $\alpha$ of the LQ survival curve, where the extent of this increase is proportional to the drug's concentration. Therefore, assuming a linear increase in the parameter $\alpha$, the additivity and radio-sensitization mechanisms yield the following survival curve
\begin{align*}
	S_{\mathrm{CRT}}(d,c)& \approx e^{-\left(\left(\alpha+ m c\right) d + \beta d^2 +\gamma c\right)}\\
	& \approx  e^{-\left(\alpha d + \beta d^2 +m c d + \gamma c \right)},
\end{align*}
where $S_{\mathrm{CRT}}$ is the fraction of cells surviving exposure to radiation dose $d$ and drug concentration $c$. The unit of the parameter $m$ is chosen such that $mc$ is in units of 1/Gy. The definition of the radiation BED model was originally motivated as the absolute value of the exponent of the LQ model divided by the linear parameter $\alpha$, that is, $-\log S_{\mathrm{RT}}(d)/\alpha$. Similarly, we assume an extended BED model for concurrent CRT as follows:
\begin{align*}
\frac{-\log S_{\mathrm{CRT}}(d,c)}{\alpha} &=d \left(1+\frac{d}{\ab} \right)+\frac{m }{\alpha}c d + \frac{\gamma}{\alpha} c \\
	&= d \left(1+\frac{d}{\ab} \right) + \ma cd + \ca c,
\end{align*}
where $\ab$, $\ma$, and $\ca$ are tissue-specific parameters.

\section{Accounting for a general normal-tissue structure} \label{sec: appendix_intermNT}
Formulation (M) is developed for treatment sites in which the dose-limiting normal tissue has a ``perfectly" parallel or serial organ structure. Therefore, the solution methods discussed in Sections \ref{sec: radiation}--\ref{sec: combined} are only applicable to these two cases. In this section, we present an approximate formulation that allows for the application of the developed solution methods to treatment sites with a general normal-tissue structure. In order to account for other organ structures the constraint in (\ref{eqn: BED}) needs to be extended to limit the generalized mean of the BED distribution in the normal tissue. More specifically, let $z_v$ represent the BED deposited in voxel $v\in V$ over the course of the treatment, that is,
\begin{align}
	z_v =  \sum_{i=1}^n B_{\mathrm{N}}\left(\delta_vd_{i},c_{i}\right). \label{eqn: BED_dist}
\end{align}
To account for a dose-limiting normal tissue with an organ structure other than parallel or serial, the constraint in (\ref{eqn: BED}) needs to be substituted with a new constraint that limits the generalized mean of the BED distribution as follows:
\begin{align}
 	\left(\frac{1}{|V|} \sum_{v  \in V} z_v^{a} \right)^{1/a} \leq \en, \label{eqn: generalizedBED} 
\end{align}
where $a\geq 1$ is a tissue-specific parameter. The parameter $a$ has been determined and reported for different organs \citep{Marks2010}. This yields an extension to (M) in which the constraint in (\ref{eqn: BED}) is replaced with (\ref{eqn: BED_dist})--(\ref{eqn: generalizedBED}). The two cases of parallel and serial organs that were previously discussed correspond to $a=1$ and $a=\infty$, respectively, which represent the average and maximum of the BED distribution as follows:
\begin{align}
	\left(\frac{1}{|V|} \sum_{v  \in V} z_v^{a} \right)^{1/a} = \begin{cases} \frac{1}{|V|} \sum_{v  \in V} z_v & a = 1; \\   \max_{v\in V}  z_v & a = \infty.  \end{cases}  \label{eqn: dualCase}
\end{align}
In order to apply our developed solution methods to the new formulation, we first use the method introduced in \cite{thieke2002} to approximate the generalized mean of the BED distribution using the convex combination of the average and maximum values. More specifically,
\begin{align}
	\left(\frac{1}{|V|} \sum_{v  \in V} z_v^{a} \right)^{1/a} \approx \lambda  \left(\frac{1}{|V|} \sum_{v  \in V} z_v \right) + \left(1-\lambda \right) \left(  \max_{v\in V}  z_v\right), \label{eqn: convex_com}
\end{align}  
where $\lambda \in \left[0,1\right]$ is the convex-combination coefficient chosen based on the $a$-parameter value (see \cite{thieke2002} for more details). Since $z_v\,(v\in V)$ increases as $\delta_v$ increases, we have $\argmax_{v\in V}  \,z_v = \argmax_{v \in V} \,\delta_v$. Let $\delta_{\max}$ represent the maximum sparing factor over all voxels $v\in V$. Substituting (\ref{eqn: convex_com}) in (\ref{eqn: generalizedBED}) and $z_v\,(v\in V)$ from (\ref{eqn: BED_dist}) yields an approximate constraint on the generalized mean of the BED distribution as follows:
\begin{align}
  \gamma  \left(\frac{1}{|V|} \sum_{v  \in V} \sum_{i=1}^n B_{\mathrm{N}}\left(\delta_vd_{i},c_{i}\right) \right) + \left(1-\gamma \right) \sum_{i=1}^n B_{\mathrm{N}}\left(\delta_{\max}d_{i},c_{i}\right) \leq \en. \label{eqn: approx}
\end{align}
Hence, replacing the constraint (\ref{eqn: BED}) of formulation (M) with (\ref{eqn: approx}) yields an approximation problem that allows for a general organ structure. In particular, it is easy to see that an argument similar to Lemma \ref{lemma: binding} applies to (\ref{eqn: approx}) requiring it to be binding at optimality. Thus, the expression in (\ref{eqn: obj}) for the objective function of (M) is also applicable to the approximation problem via substituting $\overline{\delta^2}$ with $\lambda \overline{\delta^2} + \left(1-\lambda \right) \delta^2_{\max}$ and $\bar{\delta}$ with $\lambda \bar{\delta} + \left(1-\lambda \right) \delta_{\max}$. This, in turn, allows for the application of the previously derived methods in Sections \ref{sec: radiation}--\ref{sec: combined} to the approximation problem.

\section{Proof of Theorem \ref{theorem: additive}} \label{sec: appendix_additive}
Using equations (\ref{eqn: obj_additive})--(\ref{eqn: D_additive}), the objective function of (M) can be rewritten as a univariate function, denoted by $F$, in terms of variable $b$ (the total BED delivered to the normal tissue by the chemotherapeutic agent) as follows:
\begin{align*}
	 F\left(b;n \right) = \Delta_{\mathrm{r}} D_{\mathrm{a}}\left(b;n \right) + \Delta_{\mathrm{a}} b,
\end{align*}
where the constant term $\frac{1}{\overline{\delta^2}}\frac{\abn}{\abt}\en$ is dropped from the objective function. We consider four possible cases based on the sign of $\Delta_{\mathrm{r}}$ and $\Delta_{\mathrm{a}}$ as follows:

\begin{itemize}
\item [(i)] $\Delta_{\mathrm{r}} \geq 0$ and $\Delta_{\mathrm{a}} < 0$: $F\left(\cdot;n\right)$ is maximized if $b=0$, which corresponds to no drug administration, that is, $c_i = 0 \, \left(i=1,\ldots,n \right)$. Furthermore, since $\Delta_{\mathrm{r}} \geq 0$, a standard radiation fractionation is optimal.

\item[(ii)]  $\Delta_{\mathrm{r}} < 0 $ and $\Delta_{\mathrm{a}} \geq 0$: $F\left(\cdot;n\right)$ is maximized if $b=\en$, which corresponds to no radiation administration, that is, $d_i=0 \, \left(i=1,\ldots,n\right)$.

\item[(iii)]  $\Delta_{\mathrm{r}} < 0$ and $\Delta_{\mathrm{a}} < 0$: $F(\cdot;n)$ is a convex function if $\Delta_{\mathrm{r}} < 0$  since
\begin{align*}
	F^{\prime\prime}\left(b;n \right) = -\frac{2\sfe\Delta_{\mathrm{r}}}{n\bar{\delta}^2\abn\left(1+4\sfe\left(\en - b\right) /n \bar{\delta} \abn \right)^{3/2}} >  0.
\end{align*}
Therefore, maximizing $F(\cdot;n)$  with respect to the single variable $b\in \left[0,\en\right]$ yields an optimal solution on the boundary of the feasible region, that is, $b^*=0$ or $b^*=\en$. The boundary solutions are
\begin{align*}
	F(0;n) &= \Delta_{\mathrm{r}}   \frac{-1 + \sqrt{1+4\sfe\en /n\bar{\delta}\abn}}{2\sfe/n\abn} \\
	F\left( \en;n\right) &=  \Delta_{\mathrm{a}} \en.
\end{align*}
Moreover, since $\Delta_{\mathrm{r}} < 0$, $F(\cdot;n)$ is maximized if a hypo-fractionated regimen is used, that is, n is set to the minimum number of fractions allowed $n=1$. Hence, to determine the maximum of the two values above we let
\begin{align*}
	\varrho_{\mathrm{a}} = \frac{2\sfe\en/ \abn} {-1 + \sqrt{1+4\sfe\en / \bar{\delta}\abn}},
\end{align*}
so that if $\Delta_{\mathrm{r}}/ \Delta_{\mathrm{a}}\leq \varrho_{\mathrm{a}}$, then radiation alone with a hypo-fractionation scheme is optimal. Otherwise, drug administration alone is optimal.

\item[(iv)] $\Delta_{\mathrm{r}} \geq 0$ and $\Delta_{\mathrm{a}} \geq 0$: $F\left(\cdot;n\right)$ is a concave function since
\begin{align*}
	F^{\prime\prime}\left(b;n \right) = -\frac{2\sfe\Delta_{\mathrm{r}}}{n\bar{\delta}^2\abn\left(1+4\sfe\left(\en - b\right) /n \bar{\delta} \abn \right)^{3/2}} <  0.
\end{align*}
Moreover, it has a local maximum that can be obtained by setting $F^{\prime}\left(b;n\right) = 0$, which yields
\begin{align*}
	b_{\mathrm{a}} = \en + \frac{n\abn}{4\overline{\delta^2}} \left(\bar{\delta}^2 - \left(\frac{\Delta_{\mathrm{r}}}{\Delta_{\mathrm{a}}} \right)^2\right).
\end{align*}
Clearly, $F\left(\cdot;n\right)$ is increasing for $b \leq b_{\mathrm{a}}$ and decreasing over $b > b_{\mathrm{a}}$. Thus, depending on whether or not $b_{\mathrm{a}}$ lies within the interval $\left[0,\en \right]$, we can determine the global maximum of $F$ over $\left[0,\en\right]$ as follows:
\begin{align*}
	b^* = \begin{cases}
	0 & b_{\mathrm{a}} < 0 \\
	b_{\mathrm{a}}& 0 \leq b_{\mathrm{a}} \leq \en \\
	 \en & b_{\mathrm{a}} > \en.	
\end{cases}
\end{align*}
The above conditions can be expressed in terms of $\Delta_{\mathrm{r}}/\Delta_{\mathrm{a}}$ as follows:
\begin{align*}
    b^* = \begin{cases}
    \en & \frac{\Delta_{\mathrm{r}}}{\Delta_{\mathrm{a}}} \leq \bar{\delta} \, (=\underline{\rho}_{\mathrm{a}})\\
     b_{\mathrm{a}} & \bar{\delta} \leq \frac{\Delta_{\mathrm{r}}}{\Delta_{\mathrm{a}}} \leq \bar{\delta}\sqrt{1+4\en\sfe/n\bar{\delta}\abn}\, (=\overline{\rho}_{\mathrm{a}})\\
    0 & \frac{\Delta_{\mathrm{r}}}{\Delta_{\mathrm{a}}} \geq \overline{\rho}_{\mathrm{a}}.
\end{cases}
\end{align*}
Therefore, the optimal regimen delivers a BED of $b^*$ and $\en-b^*$ using the drug and radiation administration, respectively. Finally, since $\Delta_{\mathrm{r}} \geq 0$, it is optimal to deliver radiation using a standard fractionation scheme, that is, using maximum number of fractions allowed. 
\end{itemize}

\section{Proof of Proposition \ref{prop: equalFrac}} \label{sec: appendix_stationary}
We show that there exists an optimal regimen $\big(d^*_i,c^*_i\big)\,\left(i=1,\ldots,n\right)$ to (M) in which for any given fractions $i,i^{\prime}$ such that $d^*_i, d^*_{i^{\prime}} > 0$, we have $d^*_i = d^*_{i^{\prime}}$ and $c^*_i = c^*_{i^{\prime}}$. To that end, we use the Karush-Kuhn-Tucker (KKT) optimality conditions, which are necessary conditions for local optimality if some constraint qualification is met \citep{bazaraa2006}. In particular, the linear-independence constraint qualification is satisfied at those solutions for which the constraint in (\ref{eqn: BED}) is binding, which is the case for all optimal solutions (see Lemma \ref{lemma: binding}). This requires the existence of fraction $i$ for which $d_i > 0$ since $\can = 0$ due to the lack of any additivity effect. This, in turn, leads to the linear independence of gradients of active constraints in (M). To write the KKT conditions we associate dual multiplier $\nu$ with (\ref{eqn: BED}), $\overline{\mu}_i\,(i=1,\ldots,n)$ with (\ref{eqn: drug}), $\underline{\mu}_i\,(i=1,\ldots,n)$ with non-negativity constraints on the drug's concentration in (\ref{eqn: nonneg}), and $\pi_i\,(i=1,\ldots,n)$ with non-negativity constraints on radiation doses in (\ref{eqn: nonneg}). The KKT conditions for (M) can then be expressed as follows:
\begin{align}
	&(\ref{eqn: BED})-(\ref{eqn: nonneg}) \notag \\
	&-\left(1+\mat c_i +\frac{2d_i}{\abt} \right) + \nu \left(1+\man c_i + \frac{2 \sfe  d_i}{ \abn}  \right) - \pi_i = 0 & i&=1,\ldots, n \label{kkt1}\\
	&-\mat d_i + \nu \man d_i - \underline{\mu}_i + \overline{\mu}_i  = 0  & i&=1,\ldots, n  \label{kkt2}\\
	&\nu \left(\sum_{i=1}^n \left( d_i + \frac{\sfe d_i^2}{\abn} + \man c_i d_i \right) -  \en/\bar{\delta}  \right) = 0  \label{kkt3}\\
	&\pi_id_i = 0  & i&=1,\ldots, n  \label{kkt4}\\
	&\underline{\mu}_ic_i = 0  & i&=1,\ldots, n  \label{kkt5}\\
	&\overline{\mu}_i \left(\bar{c}-c_i\right) = 0  & i&=1,\ldots, n  \label{kkt6} \\
	& \pi_i, \underline{\mu}_i, \overline{\mu}_i,  \nu \geq 0 & i&=1,\ldots, n.  \label{kkt7}
\end{align}
We show that any treatment regimen that satisfies the KKT conditions above is stationary. Consider treatment regimen $\left(d^*_i,c^*_i \right)\,\left(i=1,\ldots,n\right)$ that satisfies the KKT conditions, and suppose $d^*_i,d^*_{i^{\prime}} > 0$ for fractions $i$ and $i'$. In the following, we show that $d^*_i = d^*_{i^{\prime}}$ and $c^*_i=c^*_{i^{\prime}}$ for all possible cases:
\begin{itemize}
\item[(i)] $d^*_i, d^*_{i'} > 0$ with $c^*_i = c^*_{i'} = 0$: in that case (\ref{kkt4}) enforces $\pi^*_i =  \pi^*_{i'} = 0$. Thus using (\ref{kkt1}) and assuming $\frac{\abn}{\abt}\neq \sfe$ we have
\[ \frac{1+ 2d^*_i/ \abt}{1 + 2 \sfe d^*_i/\abn} = \frac{1 + 2d^*_{i'}/\abt}{1+ 2 \sfe d^*_{i'}\abn} \rightarrow d^*_i = d^*_{i'}.\]
\item[(ii)] $d^*_i, d^*_{i'} > 0$ with $c^*_i,c^*_{i^{\prime}} = \bar{c}$: in that case, similar to the above, using (\ref{kkt1}) and assuming that $\frac{\abn}{\abt}  \neq \sfe$ we have
\[ \frac{1+ \mat c^*_i +2d^*_i/\abt}{1 + \man c^*_i +2 \sfe d^*_i/\abn} = \frac{1+ \mat c^*_{i} + 2d^*_{i'}/\abn}{1+ \man c^*_{i'} + 2 \sfe d^*_{i'}/\abn} , c^*_{i} = c^*_{i'} = \bar{c}\rightarrow d^*_i = d^*_{i'}.\]
\item[(iii)] $d^*_i, d^*_{i'} > 0$ with $0 < c^*_i,c^*_{i'} < \bar{c}$: in that case (\ref{kkt4}), (\ref{kkt5}), and (\ref{kkt6}) enforce that $\pi^*_i =  \pi^*_{i'} = \underline{\mu}^*_i  = \underline{\mu}^*_{i'}= \overline{\mu}^*_i = \overline{\mu}^*_{i'} = 0$. Thus using (\ref{kkt2}) we have
\[ \left(\nu \man - \mat \right) d^*_i =  \left(\nu \man - \mat \right) d^*_{i'} = 0 \rightarrow \nu = \frac{\mat}{\man}. \]
Moreover, using (\ref{kkt1}) we have
\begin{align*}
\frac{1+ \mat c^*_i +2d^*_i/\abt}{1 + \man c^*_i +2 \sfe d^*_i/\abn} = \frac{1 + \mat c^*_{i'} + 2d^*_{i'}/\abt}{1+ \man c^*_{i'} + 2 \sfe d^*_{i'}/\abn}=\frac{\mat}{\man} \\ \rightarrow
d^*_{i} = d^*_{i'} = \frac{\man - \mat}{2\left(\sfe \mat /\abn - \man /\abt \right)} \rightarrow c^*_{i} = c^*_{i'}.
\end{align*}
\item[(iv)] $d^*_i, d^*_{i'} > 0$ with $c^*_i = 0$ and $c^*_{i^{\prime}} = \bar{c}$: this solution does not satisfy the KKT conditions, and we prove this claim by contradiction. Suppose such solution satisfies the KKT conditions, then using (\ref{kkt5}) and (\ref{kkt6}) we have $ \overline{\mu}_i =  \underline{\mu}_{i^{\prime}} = 0$. Furthermore, using (\ref{kkt2}) we have
\[\underline{\mu}_i= d_i \left(-\mat + \nu \man \right),\, -\overline{\mu}_{i^{\prime}} = d_{i^{\prime}} \left(-\mat + \nu \man \right). \]
However, $\underline{\mu}_i$ and $\overline{\mu}_{i^{\prime}}$ in the above cannot simultaneously satisfy nonnegativity conditions required by (\ref{kkt7}) unless $\nu=\mat/\man$. However, we showed in case (iii) that $\nu=\mat/\man$ would imply $d^*_i = d^*_{i^{\prime}}$ and $c^*_i = c^*_{i^{\prime}}$, which is a contradiction.
\item[(v)] $d^*_i, d^*_{i'} > 0$ with either (v-i) $c^*_i = 0$ and $0 < c^*_{i^{\prime}} < \bar{c}$ or (v-ii) $c^*_i = \bar{c}$ and $0 < c^*_{i^{\prime}} < \bar{c}$. We prove that neither (v-i) nor (v-ii) satisfy the KKT conditions by contradiction. We only discuss the case of (v-i) and a similar argument applies to (v-ii). Suppose the solution satisfied the KKT conditions, then (\ref{kkt4}), (\ref{kkt5}), and (\ref{kkt6}) would enforce $\pi^*_i =  \pi^*_{i'}  = \overline{\mu}^*_i =\underline{\mu}^*_{i'}=  \overline{\mu}^*_{i'} = 0$. Thus using (\ref{kkt2}) we have
\[ \left(\nu \man - \mat \right) d^*_{i'} = 0 \rightarrow \nu = \frac{\mat}{\man}.\]
Substituting $\nu$ in (\ref{kkt1}) yields
\begin{align*}
\frac{1 +2d^*_i/\abt}{1 +2 \sfe d^*_i/\abn} = \frac{1 + \mat c^*_{i'} + 2d^*_{i'}/\abt}{1+ \man c^*_{i'} + 2 \sfe d^*_{i'}/\abn}=\frac{\mat}{\man} \\ \rightarrow
d^*_{i} = d^*_{i'} = \frac{\man - \mat}{2\left(\sfe \mat /\abn - \man /\abt \right)} \rightarrow c^*_{i'} = 0
\end{align*}
which is a contradiction. 
\end{itemize}
Therefore, it is shown that for all treatment regimens that satisfy the KKT conditions, if there exist $i,i^{\prime}$ such that $d^*_i , d^*_{i^{\prime}} > 0$, then we have $d^*_i = d^*_{i^{\prime}}$ and $c^*_i = c^*_{i^{\prime}}$.

\section{Proof of Theorem \ref{theorem: synergistic}} \label{sec: appendix_synergistic}
Considering only stationary regimens, the objective function of (M) in equation (\ref{eqn: obj_synergistic}) can be expressed as a univariate function, denoted by $G$, in terms of variable $c$ (the stationary drug concentration) as follows:
\[ G\left(c;n\right)= \left(\Delta_{\mathrm{r}} + \man \Delta_{\mathrm{s}} c \right) D_{\mathrm{s}}\left(c;n\right),\]
where the constant term $\frac{1}{\overline{\delta^2}}\frac{\abn}{\abt}\en$ is dropped from the objective function. $G(\cdot;n)$ has only a single extremum
\[c_{\mathrm{s}}=\frac{\overline{\rho}_{\mathrm{s}}^2-\left(\Delta_{\mathrm{r}}/\Delta_{\mathrm{s}}\right)^2}{ 2 \man \left(\Delta_{\mathrm{r}} / \Delta_{\mathrm{s}} - 1\right)},\]
which is obtained by setting $G^{\prime}(\cdot;n) = 0$ and thus is a unimodal function. To determine the maximum of $G(\cdot;n)$ over the interval $\left[0,\bar{c}\right]$, we evaluate $G^{\prime}(\cdot;n)$ at the end points of the feasible region $\left[0,\bar{c}\right]$. We consider the following four possible cases depending on the sign of $\Delta_{\mathrm{r}}$ and $\Delta_{\mathrm{s}}$:

\begin{itemize}
\item [(i)] $\Delta_{\mathrm{r}}, \Delta_{\mathrm{s}} \geq 0$:
$G^{\prime}(0;n) \geq 0$ if and only if
\[ \frac{\Delta_{\mathrm{r}} }{ \Delta_{\mathrm{s}}} \leq  \sqrt{1+4\sfe \en / n\bar{\delta}\abn}\; (=\overline{\rho}_{\mathrm{s}}),\] 
similarly, $G^{\prime}\left(\bar{c};n\right) \geq 0$ if and only if
\[\frac{\Delta_{\mathrm{r}} }{ \Delta_{\mathrm{s}}} \leq   \sqrt{\left(1+\man \bar{c}\right)^2+4\sfe\en / n \bar{\delta}\abn} - \man\bar{c}\; (=\underline{\rho}_{\mathrm{s}}).\]
It is then easy to verify that
\begin{align}
	 \sqrt{\left(1+\man \bar{c}\right)^2+4\sfe\en / n \bar{\delta}\abn} - \man\bar{c} \leq  \sqrt{1+4\sfe \en / n\bar{\delta}\abn}. \label{eqn: bounds}
\end{align}
Based on the above, if
\[\frac{\Delta_{\mathrm{r}} }{ \Delta_{\mathrm{s}}} \leq  \underline{\rho}_{\mathrm{s}},\]
then $G(\cdot;n)$ is increasing over $\left[0,\bar{c}\right]$ and $c^*=\bar{c}$. Next, consider the case for which $G^{\prime}(0;n) \geq 0$ and $G^{\prime}\left(\bar{c};n\right) \leq 0$, this corresponds to
\begin{align*}
	 \underline{\rho}_{\mathrm{s}} \leq \frac{\Delta_{\mathrm{r}} }{ \Delta_{\mathrm{s}}} \leq  \overline{\rho}_{\mathrm{s}}.
\end{align*}
In this case the extremum $ c_{\mathrm{s}} \in [0,\bar{c}]$ is a local maximum and thus the global optimal solution. We next consider the case in which $G^{\prime}(0;n) \leq 0$ and $G^{\prime}\left(\bar{c};n\right) \geq 0$. However, this is not possible since it requires $\underline{\rho}_{\mathrm{s}} \geq  \overline{\rho}_{\mathrm{s}}$, which contradicts (\ref{eqn: bounds}). Finally, we consider $G^{\prime}(0;n)  \leq 0$ and $G^{\prime}\left(\bar{c};n\right) \leq 0$, which corresponds to
\begin{align*}
	\frac{\Delta_{\mathrm{r}} }{ \Delta_{\mathrm{s}}} & \geq   \overline{\rho}_{\mathrm{s}}.
\end{align*}
In this case $G^{\prime}(\cdot;n)$ is negative and $G(\cdot;n)$ is decreasing over $[0,\bar{c}]$. Hence, the optimal solution is at $c^*=0$. In addition, since $\Delta_{\mathrm{r}}, \Delta_{\mathrm{s}} \geq 0$,  we have $\Delta_{\mathrm{r}}+\Delta_{\mathrm{s}}c \geq 0$ for $c \in \left[0,\bar{c}\right]$ and thus a standard fractionation regimen is optimal.

\item[(ii)]  $\Delta_{\mathrm{r}} \leq 0$ and $\Delta_{\mathrm{s}} \geq 0$: similar to case (i) $ G^{\prime}\left( c;n\right) \geq 0$ for $c\in \left[0,\bar{c}\right]$ if and only if
\[\frac{\Delta_{\mathrm{r}}}{\Delta_{\mathrm{s}}} \leq   \sqrt{\left(1+\man c\right)^2+4\sfe\en / n \bar{\delta}\abn} - \man c.\]
This condition is met over the interval $[0,\bar{c}]$ because $\Delta_{\mathrm{r}} /\Delta_{\mathrm{s}} \leq 0$ and the right-hand-side of the above inequality is always positive. Therefore, $G$ is increasing over $[0,\bar{c}]$ and the optimal solution is at $c^*=\bar{c}$. Additionally, when
\[ \frac{\Delta_{\mathrm{r}} }{\Delta_{\mathrm{s}}} < -\man \bar{c},\]
then $\Delta_{\mathrm{r}} +  \man\Delta_{\mathrm{s}}\bar{c} < 0$ and thus a hypo-fractionated radiation regimen is optimal. On the contrary, when the above condition is not satisfied, 
then $\Delta_{\mathrm{r}} +  \man\Delta_{\mathrm{s}}\bar{c} \geq 0$ and a standard fractionation regimen is optimal.

\item[(iii)]  $\Delta_{\mathrm{r}} \leq 0$ and $\Delta_{\mathrm{s}} \leq 0$: $G^{\prime}(0;n) \geq 0$ if and only if
\[ \frac{\Delta_{\mathrm{r}} }{ \Delta_{\mathrm{s}}} \geq   \overline{\rho}_{\mathrm{s}},\] 
and $G^{\prime}\left(\bar{c};n\right) \geq 0$ if and only if
\[\frac{\Delta_{\mathrm{r}} }{ \Delta_{\mathrm{s}}} \geq  \underline{\rho}_{\mathrm{s}}.\]
Therefore, it is not possible to have $G^{\prime}(0;n) \geq 0$ and $G^{\prime}(\bar{c};n) \leq 0$ since in that case we would have $\underline{\rho}_{\mathrm{s}} \geq  \overline{\rho}_{\mathrm{s}}$, which is in contradiction with inequality (\ref{eqn: bounds}). Thus, there is no local maximum within the interval $\left(0,\bar{c}\right)$ and the optimal solution must lie on the boundary of $[0,\bar{c}]$, that is, $c^* = 0$ or $c^* = \bar{c}$. More specifically, the global maximum is
\begin{align*}
G^* &= \max \Big\{G(0;n),G(\bar{c};n) \Big\} \\
	& = \max\Big\{\Delta_{\mathrm{r}} D_{\mathrm{s}}\left(0;n\right), \left(\Delta_{\mathrm{r}} + \man\Delta_{\mathrm{s}} \bar{c}\right) D_{\mathrm{s}}\left(\bar{c};n\right) \Big\},	
\end{align*}
which can be determined using the following condition: if
\[\frac{\Delta_{\mathrm{r}}}{\Delta_{\mathrm{s}}} \leq \frac{\man\bar{c}}{  D_{\mathrm{s}}\left(0;n\right) / D_{\mathrm{s}}\left(\bar{c};n\right) - 1} \;(=\varrho_{\mathrm{s}}), \]
then $c^* = 0$; otherwise, $c^* = \bar{c}$. Finally, since $\Delta_{\mathrm{r}} + \man \Delta_{\mathrm{s}}c \leq 0$ for $c \geq 0$, a hypo-fractionated radiation regimen is always optimal, that is, $n=1$.

\item[(iv)] $\Delta_{\mathrm{r}} \geq 0$ and $\Delta_{\mathrm{s}} \leq 0$:
for any given $c \in \left[0,\bar{c}\right]$ we have $G^{\prime}\left(c;n\right) \geq 0$ if and only if
\[\frac{\Delta_{\mathrm{r}} }{ \Delta_{\mathrm{s}}} \geq   \sqrt{\left(1+\man c\right)^2+4\sfe\en / n \bar{\delta}\abn} - \man c \qquad \qquad \forall c \in \left[0,\bar{c}\right].\]
However, since $\Delta_{\mathrm{r}}/\Delta_{\mathrm{s}}\leq 0$ and the right-hand-side of the above inequality is positive, the above condition cannot be satisfied and $G^{\prime}\left(c;n \right)\leq 0$ for $c \in \left[0,\bar{c}\right]$.
Therefore, $G\left(c;n \right)$ is decreasing over the interval $\left[0, \bar{c}\right]$ and $c^* = 0$ is the global maximum. Moreover, at $c=0$,  $\Delta_{\mathrm{r}} + \man \Delta_{\mathrm{s}}c = \Delta_{\mathrm{r}} > 0 $ and thus a standard fractionation regimen is optimal.  
\end{itemize}

\section{Trade-off between target and normal-tissue BED} \label{sec: appendix_tradeoff}
The second derivative of $H$ with respect to $b$ is
\[	H^{\prime \prime} \left(b;n\right) = \frac{-2\sfe\Delta_{\mathrm{r}}}{\abn \bar{\delta}^2 \left( 1+4\sfe b / n\bar{\delta}\abn\right)^{3/2}}.\]
Therefore, it is easy to see that $H$ is convex when $\Delta_{\mathrm{r}} < 0$ and concave otherwise.

\end{appendices}
\end{document}